\documentclass{article}

\usepackage{PRIMEarxiv}
\usepackage{times}
\usepackage{soul}
\usepackage{url}
\usepackage[hidelinks]{hyperref}
\usepackage[utf8]{inputenc}
\usepackage[small]{caption}
\usepackage{graphicx}
\usepackage{amsmath}
\usepackage{amsthm}
\usepackage{booktabs}
\usepackage{algorithm}
\usepackage{algorithmic}
\usepackage{color}
\usepackage{amsmath, amsthm, amsfonts, mathrsfs}
\usepackage{bm}
\usepackage{multirow}
\usepackage{booktabs}
\usepackage{mathtools}
\usepackage{appendix}
\usepackage[sort&compress,numbers]{natbib}
\DeclareMathOperator*{\argmin}{arg\,min}

\usepackage{graphicx}
\usepackage{algorithm}
\usepackage{algorithmic}
\usepackage{hyperref}
\usepackage{bookmark}
\usepackage{xcolor}
\usepackage{amsmath, amsfonts, mathrsfs}
\usepackage{bm}
\usepackage{multirow}
\usepackage{mathtools}
\usepackage{cases}
\usepackage{array}
\usepackage{makecell}
\usepackage[noabbrev]{cleveref}

\newcommand{\myscale}{0.4}

\newcommand{\tvfud}{\mu}
\newcommand{\tieset}{F}

\renewcommand{\Pr}{\mathrm{P}}

\newtheorem{lemma}{Lemma}

\newtheorem{theorem}{Theorem}
\newtheorem{proposition}{Proposition}
\newtheorem{observation}{Observation}

\newtheorem{remark}{Remark}

\title{Facility Location Games Beyond Single-Peakedness: the Entrance Fee Model}

\author{
  Mengfan Ma\\
  School of Computer Science \\and Engineering, University of \\Electronic Science and Technology of China\\
  \texttt{mengfanma1@gmail.com} \\
  \And 
  Mingyu Xiao\\
  School of Computer Science \\and Engineering, University of \\Electronic Science and Technology of China\\
  \texttt{myxiao@uestc.edu.cn}\\
  \And
  Tian Bai \\
  School of Computer Science \\and Engineering, University of \\Electronic Science and Technology of China\\
  \texttt{tian.bai.cs@outlook.com} \\
  \And
  Bakh Khoussainov\\
  School of Computer Science \\and Engineering, University of \\Electronic Science and Technology of China\\
  \texttt{bmk@uestc.edu.cn}\\
}

\begin{document}

\maketitle

\begin{abstract}
      The facility location game has been studied extensively in mechanism design. In the classical model, each agent's cost is solely determined by her distance to the nearest facility. In this paper, we introduce a novel model where each facility charges an entrance fee. Thus, the cost of each agent is determined by both the distance to the facility and the entrance fee of the facility. In our model, the entrance fee function is allowed to be an arbitrary function, causing agents' preferences may no longer be single-peaked anymore: This departure from the classical model introduces additional challenges. We systematically delve into the intricacies of the model, designing strategyproof mechanisms with favorable approximation ratios. Additionally, we complement these ratios with nearly-tight impossibility results. Specifically, for one-facility and two-facility games, we provide upper and lower bounds for the approximation ratios given by deterministic and randomized mechanisms with respect to utilitarian and egalitarian objectives. 
\end{abstract}

\section{Introduction}
In the one-dimensional facility location problem, agents are located on the real line, and a planner's goal is to build one or more facilities on the line to serve the agents. The cost of an agent is her distance to the nearest facility. The problem asks to locate facilities to minimize the total cost of all agents (the utilitarian objective) or the maximum cost among all agents (the egalitarian objective). It is well-known that both optimization problems can be solved in polynomial time \cite{love1976adynamic}. 
Over the past decade \cite{procaccia2009approximate,lu2010asymptotically,alon2010strategyproof,cai2016facility,aziz2020facility,wang2023multi}, these problems have undergone intensive study from the perspectives of mechanism design and game theory.
A key conversion in the models is that now each agent becomes strategic and may misreport her position to decrease her cost. These new problems are called the \emph{facility location games}, which require to design \emph{mechanisms} that elicit the true positions of agents and output facility locations to (approximately) minimize the total or maximum cost.

\citet{moulin1980strategy} provides a complete characterization of strategyproof mechanisms for the facility location game on line where agents have single-peaked preferences, with no consideration for the optimization of the objective function.
The seminal work of Procaccia et al.\! \cite{procaccia2009approximate} initiates the study of approximate mechanism design without money for the facility location game.
They study strategyproof mechanisms for one-facility and two-facility games through the lens of \emph{approximation ratio} of the objective, which provides a way to quantify the fundamental tension between truthfulness and optimality.
Since then, the facility location game has become one of the main grounds for approximate mechanism design without money and has attracted numerous follow-up research~\cite{lu2009tighter,lu2010asymptotically, alon2010strategyproof, fotakis2014power}.
In particular, \citet{fotakis2014power} show that no deterministic strategyproof mechanism can achieve a bounded approximation ratio when there are more than two facilities.
The current status of the classical facility location games on the real line with one or two facilities is summarized in Table \ref{tb_class_results}. We can observe that all the bounds are tight or asymptotically tight.

After the introduction of the classical model, many variants have been studied flourishingly to accommodate more practical scenarios.
The recent survey by \citet{Chan2021mechanism} depicts state of the art. Here, we mention some of the models: obnoxious facility games where every agent wants to stay away from the facility \cite{cheng2013strategy, feigenbaum2015strategyproof};
heterogeneous facility games where the acceptable set of facilities for each agent could be different \cite{serafino2015truthful, li2019strategyproof,deligkas2022heterogeneous,Kanellopoulos2022ondiscrete};
the games with different objectives or purposes \cite{cai2016facility, zhou2022strategyproof};
the game with capacitated facilities \cite{aziz2020facility,aziz2020capacity, walsh2022strategy}; the game with strategic facilities \cite{li2021budgeted}; the game where agents can hide their locations \cite{hossain2020thesurprising}; extensions to trees, cycle and networks \cite{schummer2002strategy, alon2010strategyproof};
and so on.

\smallskip
\noindent
\textbf{Facility location games with entrance fees.}
In all the above models, the cost of an agent is measured by her distance to the closest facility. This cost can be considered as the \emph{travel fee}. In many real-life scenarios, except for the travel fee, the agent may also need to pay the facility a service or entrance fee, such as tickets for swimming pools and museums.
The entrance fees may differ for facilities in different locations. An immediate example is building a facility downtown would be more expensive than in the suburbs.
The entrance fee of the facility is decided by the building cost and thus also decided by the location where the facility is built. Another example is that the environment around the facility may {{have}} an impact on the entrance fee. For example, the entrance fee of a facility in a popular scenic spot can be higher than the entrance fee of the same facility in a desolate place.
Various non-geographical settings can be accommodated within the location-dependent entrance fee model. Two illustrative examples include:
(1) In a scenario where voters exhibit a single-peaked preference over potential candidates, after the election, the winning candidates implement their own tax policy on the voters. This tax can be conceptualized as an entrance fee, with the act of selecting the winning candidates akin to choosing a facility location. 
(2) Consider the context of roommates residing in a college dormitory faced with the task of determining the air conditioner's temperature. The associated electricity bill, contingent upon the temperature setting, is subsequently divided equally among the roommates. Different temperature settings may result in varying division of the bill, representing an entrance fee determined by the chosen temperature.

\begin{table*} [t]
    \small
    \centering
    \begin{tabular}
        {|c|c|c|c|c|c|}
    \hline
    \multicolumn{2}{|c|}{\multirow{2}[4]{*}{}} & \multicolumn{2}{c|}{{Total cost}} & \multicolumn{2}{c|}{{Maximum cost}} \\
    \cline{3-6}\multicolumn{2}{|c|}{} & \multicolumn{1}{c|}{Upper bounds} & \multicolumn{1}{c|}{Lower bounds} & \multicolumn{1}{c|}{Upper bounds} & \multicolumn{1}{c|}{Lower bounds} \\
    \hline
    \multicolumn{1}{|c|}{\multirow{2}[4]{*}
        {\thead{{Deterministic} \\ {Mechanisms}}}
    }
    &
        \thead{{One-facility}}
    &
        \thead{$1$ \cite{procaccia2009approximate}}
    &
		\thead{$1$ \cite{procaccia2009approximate}}
    &
    \multirow{2}[4]{*}{
        \thead{$2$ \cite{procaccia2009approximate}}
    }
    &
    \multirow{2}[4]{*}
    {
        \thead{$2$ \cite{procaccia2009approximate}}
    } \\
    \cline{2-4}
    &
        \thead{{Two-facility}}
    &
        \thead{\(n-2\) \cite{procaccia2009approximate}}
    &
        \thead{\(n-2\) \cite{fotakis2014power}}
    &
    &  \\
    \cline{1-4}\cline{5-6}\multicolumn{1}{|m{60.001pt}<{\centering}|}{\multirow{2}[4]{*}
    {
        \thead{{Randomized} \\ {Mechanisms}}
    }}
    &
        \thead{{One-facility}}
    &
        \thead{\(1\) \cite{procaccia2009approximate}}
    &
        \thead{$1$ \cite{procaccia2009approximate}}
    &
        \thead{$\frac{3}{2}$ \cite{procaccia2009approximate}}
    &
    \multirow{2}[4]{*}
    {
        \thead{$\frac{3}{2}$ \cite{procaccia2009approximate}}
    } \\
    \cline{2-5}
    &
        \thead{{Two-facility}}
    &
        \thead{\(4\) \cite{lu2010asymptotically}}
    &
        \thead{\(1.045\) \cite{lu2009tighter}}
    &
        \thead{$\frac{5}{3}$ \cite{procaccia2009approximate}}
    &  \\
    \hline
    \end{tabular}
    \caption{Upper and lower bounds for the approximation ratios of strategyproof mechanisms in the classical model.}
    \label{tb_class_results}
\end{table*}

\begin{table*} [t]
    \small
    \centering
    \begin{tabular}
		{|c|c|c|c|c|c|}
    \hline
    \multicolumn{2}{|c|}{\multirow{2}[4]{*}{}} & \multicolumn{2}{c|}{Total cost} & \multicolumn{2}{c|}{Maximum cost} \\
    \cline{3-6}\multicolumn{2}{|c|}{} & \multicolumn{1}{c|}{Upper bounds} & \multicolumn{1}{c|}{Lower bounds} & \multicolumn{1}{c|}{Upper bounds} & \multicolumn{1}{c|}{Lower bounds} \\
    \hline
    \multicolumn{1}{|c|}{\multirow{2}[4]{*}
          {
            \thead{Deterministic \\ Mechanisms}
          }
    }
    &
    \thead{One-facility}
    &
    \thead{$3$}
    &
    \thead{$3$}
    &
    \multirow{4}[8]{*}
    {
    \thead{$3$}
    }
    &
    \multirow{2}[4]{*}
    {
    \thead{$3$}
    } \\
    \cline{2-4}
    &
    \thead{Two-facility}
    &
    \thead{\(n-1\)}
    &
    \thead{\(n-2\)}
    &
    &  \\
    \cline{1-4}\cline{6-6}\multicolumn{1}{|c|}{\multirow{2}[4]{*}
    {
      \thead{Randomized \\ Mechanisms}
    }}
    &
    \thead{One-facility}
    &
    \thead{$3 - \frac{2}{n}$}
    &
    \thead{$2$}
    &
    &
    \thead{$2$}\\
    \cline{2-4} \cline{6-6}
    &
    \thead{Two-facility}
    &
    \thead{\(n-1\)}
    &
    \thead{\(2\)}
    &
    &
    \thead{$2$} \\
    \hline
    \end{tabular}%
    \caption{Simplified upper and lower bounds for the approximation ratios of strategyproof mechanisms in our model.}
    \label{tb_our_results_simple}
\end{table*}

\begin{table*} [t]
    \centering
    \resizebox{\textwidth}{!}{
    \begin{tabular}
		{|m{35.001pt}<{\centering}|m{2.602em}<{\centering}|m{6.003em}<{\centering}|m{11.504em}<{\centering}|m{9.605em}<{\centering}|m{13.606em}<{\centering}|}
    \hline
    \multicolumn{2}{|c|}{\multirow{2}[4]{*}{}} & \multicolumn{2}{c|}{{Total cost}} & \multicolumn{2}{c|}{{Maximum cost}} \\
    \cline{3-6}\multicolumn{2}{|c|}{} & \multicolumn{1}{c|}{Upper bounds} & \multicolumn{1}{c|}{Lower bounds} & \multicolumn{1}{c|}{Upper bounds} & \multicolumn{1}{c|}{Lower bounds} \\
    \hline
    \multicolumn{1}{|m{35.001pt}<{\centering}|}{\multirow{2}[4]{*}
          {
            \thead{{Deter-}\\ {ministic} \\ \\ {Mecha-} \\ {nisms}}
          }
    }
    &
    \thead{{One-}\\ {facility}}
    &
    \thead{${3 - \frac{4}{r_e+1}}$ \\
    ({ {Thm. \ref{thm_1f_tc_apx}}})}
    &
    \thead{$3 - \frac{16}{r_{e} + 5 + \sqrt{r_{e}^{2} + 10r_{e} - 7}}$\\
    ({ {Thm. \ref{thm_1f_tc_d_lb_3re}}})}
    &
    \multirow{4}[8]{*}
    {
    \thead{\\ \\ \\ \\ \\ \\ \\ \\ \\
    $2$, \hspace{2.3em}if $r_e\le 2$ \\
    $3-\frac{2}{r_e}$, if $r_e > 2$\\ \\
    ({Thm. \ref{thm_1f_mc_apx} and Prop. \ref{prop_2f_mc_apx}})
    }
    }
    &
    \multirow{2}[4]{*}
    {
          \thead{\\ \\ $2$, \hspace{9.78em}if $r_e\!\le 6$ \\
          $3-\frac{28}{\sqrt{r_e^{2}+20r_e-12}+r_e+10}$, if $r_e \!> 6$\\
		  \\
           ({{ Thms. \ref{thm_1f_mc_lb_d_2} and \ref{thm_1f_mc_d_lb_3re}}})\\
           ({ {Props. \ref{prop_2f_mc_lb_2} and \ref{prop_2f_mc_d_lb_3re}}})}
    } \\
    \cline{2-4}
    &
    \thead{{Two-}\\ {facility}}
    &
    \thead{\(n-1\)  \\
    ({ {Prop. \ref{prop_2f_tc_apx}}})}
    &
    \thead{\\ \(n-2\text{,\hspace{62.5pt}if } r_e\!=\!1\)
	\\
	\(3\! -\! \frac{24}{r_{e} + 7 + \sqrt{r_{e}^{2} \!+ 14r_{e} + 1}} \text{, if } r_e\!>\!1\)\\ \\
	(\cite{fotakis2014power})\\
	({Prop.} \ref{prop_2f_tc_d_lb_3re})
	}
    &
    &  \\
    \cline{1-4}\cline{6-6}\multicolumn{1}{|m{35.001pt}<{\centering}|}{\multirow{2}[4]{*}
    {
      \thead{{Rand-}\\ {omized} \\ \\ {Mecha-} \\ {nisms}}
    }}
    &
    \thead{{One-}\\ {facility}}
    &
    \thead{$3 - \frac{2}{n}$  \\
    ({ {Thm. \ref{thm_1f_tc_r_ub}}})}
    &
    \thead{$\frac{\sqrt{2} + \!1}{2} \!-\! \frac{1}{(\!4+\!2\sqrt{2}) r_{e} \!-\! 2}$, if  \(r_e\!<\! +\infty\) \\
    $2$, \hspace{7.39em}if \(r_e\!=\!+\infty\) \\
    ({Thms. \ref{thm_1f_tc_r_lb_2} and \ref{thm_1f_tc_r_lb_re}})}
    &
    &
	\thead{\\
		\(\frac{24\sqrt{2}+27}{47} - \frac{4}{(60\sqrt{2} + 97)r_{e} - (54\sqrt{2} + 92)}\)\vspace{3pt}, \\ \hspace{90pt} if \(r_e< +\infty\) \vspace{3pt}\\
		$2$, \hspace{77pt} if \(r_e=+\infty\) \\
		({ {Thms. \ref{thm_1f_mc_r_lb_re} and \ref{thm_1f_mc_r_lb_2}}})} \\
    \cline{2-4} \cline{6-6}
    &
    \thead{{Two-}\\ {facility}}
    &
    \thead{\(n-1\) \\
           ({ {Prop. \ref{prop_2f_tc_apx}}})}
    &
    \thead{\\ \(1.045\), \hspace{14 pt}if \(r_e=1\)\\
	\(2\), \hspace{19pt}if \(r_e=+\infty\) \vspace{4pt}\\
    (\cite{lu2010asymptotically})\\
	({Prop.} \ref{prop_2f_tc_r_lb_2})}
    &
    &
	\thead{\\
		$2$ if \(r_e=+\infty\) \\
		({ {Prop. \ref{prop_2f_mc_r_lb}}})} \\
    \hline
    \end{tabular}%
    }
    \caption{Upper and lower bounds for the approximation ratios of strategyproof mechanisms in our model, parameterized by \(r_e\).}
    \label{tb_our_results}
\end{table*}

All of the above motivate us to initiate the study of the \emph{facility location games with entrance fees}.
In our setting, the planner needs to locate a given number of facilities on the real line $\mathbb R$ to serve agents on the line.
Each facility, once located, has an entrance fee determined by its location. The cost of an agent is the sum of the travel fee (distance to the facility) and the entrance fee of the facility. Each agent will use one facility at a minimum cost.
\footnote{In this paper, we make the assumption that facilities are homogeneous in functionality, and agents have a unit-demand for the services provided by these facilities. The consideration of heterogeneous facilities falls outside the scope of this paper, and we direct interested readers to \cite{li2019strategyproof,Kanellopoulos2022ondiscrete,deligkas2022heterogeneous}.}
The position of each agent is private information. We want to design strategyproof mechanisms that guarantee that the agents report their true positions and locate the facilities based on the reports such that either the total or the maximum cost approximates the optimal value of the corresponding optimization problem as closely as possible. Figure \ref{fig_ex_ef} illustrates the concept of the entrance fee function.

\begin{figure} [h]
      \centering
      \includegraphics[width=\myscale\linewidth]{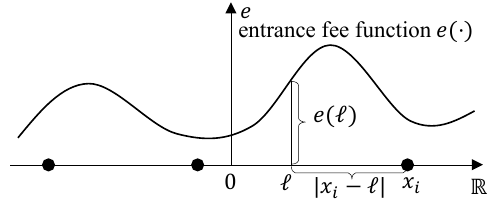}
      \caption{The concept of the entrance fee function. The three black dots are agents. The vertical axis stands for the value of the entrance fee. The curve depicts an entrance fee function \(e: \mathbb{R}\rightarrow \mathbb{R}_{\ge 0}\!\cup\! \{+\infty\}\). If the facility is located at \(\ell\), the entrance fee is \(e(\ell)\) and the cost of agent \(i\) is \(\tvfud|\ell-x_i|+e(\ell)\), where \(\tvfud\) is the cost per unit-of-distance.} 
      \label{fig_ex_ef}
\end{figure}

\smallskip
\noindent
\textbf{Contribution.}
Our primary conceptual contribution lies in introducing the concept of an \emph{entrance fee function} to facility location games. In our model, the entrance fee function is arbitrary, forming a dynamic part of the input. This captures a broader range of real-life scenarios where facilities incur location-dependent service fees for their customers. Furthermore, the notion of entrance fees can be applied to various previous variants of facility location problems, expanding the scope of research in this area.

However, the arbitrariness of the entrance fee function introduces new challenges in designing strategyproof mechanisms. Agent preferences may no longer adhere to single-peakedness \cite{moulin1980strategy, Barbera1993generalized}, and standard mechanisms for the classical model cannot be directly extended to our setting while preserving strategyproofness. To address this, we consider the optimal location for each agent, minimizing her cost if a facility is built at this optimal location. We design mechanisms over the space of these optimal locations, showing that they satisfy a crucial monotonicity property for strategyproofness. Our results carry an important theoretical implication:

\begin{center}
    \begin{tabular}{l}
        \emph{Classical mechanisms can be extended to a preference domain beyond single-peakedness} \\
        \emph{while maintaining strategyproofness.} 
    \end{tabular}
\end{center}

To analyze the approximation ratios of our mechanisms, we introduce the notion of the max-min ratio of the entrance fee function \(e\), denoted by \(r_e\). Intuitively, \(r_e\) quantifies the "turbulence" of \(e(\cdot)\), defined as the ratio of the maximum value to the minimum value of \(e(\cdot)\). Our analysis leads to approximation ratios parameterized by \(r_e\), as summarized in \Cref{tb_our_results}. According to \Cref{tb_our_results}, we can get a simplified summary without the parameter \(r_e\), as in \Cref{tb_our_results_simple}.
As one might expect, the approximation ratio increases with \(r_e\). What may be surprising is that the "turbulence" of \(e(\cdot)\) has a limited impact on the approximation ratio of our mechanisms, compared to their classical counterparts.
Notably, for the one-facility game, the approximation ratio for the total cost of the extended \textsc{Median} mechanism only deteriorates from \(1\) to \(3\) as \(r_e\) approaches \(+\infty\). More evidence can be found by comparing \Cref{tb_class_results,tb_our_results}.
When the entrance fee is consistently 0, our model reduces to the classical model. Therefore, our mechanisms must also encompass those of the classical model: By letting \(r_e=1\), one can verify that the ratios in \Cref{tb_our_results} match those in \Cref{tb_class_results}.

Moreover, we complement the proposed mechanisms with tight or nearly tight lower bounds, also parameterized by \(r_e\). While lower bounds for the classical model are applicable in our model, given that the classical model is a special case of ours, we establish new and often improved bounds in most cases. The techniques employed for deriving these lower bounds are novel and distinct from those used in the standard model due to the presence of entrance fee functions.

The rest of the paper is organized as follows. Section \ref{sect_model} presents the formal definitions of our model. In Section \ref{sect_stct_prop}, we derive some useful {{technical}} properties.
Sections \ref{sect_1f_tc} and \ref{sect_1f_mc} present our main results for the one-facility game with total and maximum cost objectives, respectively.
In Section \ref{sect_2f}, we extend our results {{from}} one-facility games to two-facility games.
In Section \ref{sect_con}, we draw conclusions and point out some possible directions for future work.

\section{Model} \label{sect_model}

Let $N= \{1,\dots, n\}$ be a set of \emph{agents}  on real line $\mathbb R$. Let $x_i$ be the position of agent $i$.
The \emph{agent position profile} is {{the vector}} \(\mathbf{x}=(x_1, \dots, x_n)\).
We assume agents are ordered such that $x_1 \le \dots \le x_n$.
We need to build  $m$ \emph{facilities} on $\mathbb R$ to serve the agents. {{If}} we put facility $j$ at location $\ell_j$, then this gives us the \emph{facility location profile} \(\bm{\ell}=(\ell_1, \dots, \ell_m)\).
A singleton $\{i\}$ may be simply written as $i$.

If an agent \(i\) is served by a facility \(j\), the agent incurs a travel fee of \(\tvfud \cdot |x_i - \ell_j|\), where \(\tvfud > 0\) represents the \emph{cost per unit-of-distance}. Additionally, each facility imposes an \emph{entrance fee}, determined by its location. Formally, an \emph{entrance fee function} \(e: \mathbb{R}\rightarrow \mathbb{R}_{\ge 0} \cup \{+\infty\}\) is defined\footnote{Note that, in our definition, a location's entrance fee can be \(+\infty\). This is justified in scenarios where such a location exists: either the planner will never build a facility at it, or no agents will select it. We also establish the relation \(x < +\infty\) for any \(x \in \mathbb{R}_{\ge 0}\).}. The entrance fee of the facility located at \(\ell\) is denoted by \(e(\ell)\). Thus, when an agent \(i\) selects facility \(j\), the total cost incurred by the agent is the sum: \(\tvfud |x_i - \ell_j| + e(\ell_j)\). 
Hereafter, we normalize \(\tvfud\) to be \(1\) for simplicity, and it is important to note that all our results remain applicable for any other positive value of \(\tvfud\).\footnote{In this paper, we make the assumption that the travel fee is a linear function of the distance. While it is undoubtedly intriguing to explore more complex functions, such considerations are deferred to future work.}

For an entrance fee function \(e\), let \(e_{\mathrm{max}}=\max_{x\in \mathbb R} e(x) \) and \(e_{\mathrm{min}}=\min_{x\in \mathbb R} e(x) \).
The \emph{max-min ratio} of \(e\), denoted by $r_e$, is defined as \(r_e=1\) if \(e_{\mathrm{min}}=e_{\mathrm{max}}\);  \(r_e=+\infty\) if \(e_{\mathrm{min}}=0\) and \(e_{\mathrm{max}}>0\); \(r_e=e_{\mathrm{max}} /e_{\mathrm{min}}\) otherwise.

Each agent always selects a  facility that minimizes the sum of her travel fee and entrance fee.
So, we define the \emph{cost} of agent $i$ for a given facility location profile \(\bm{\ell}\) as $\mathrm{cost}(x_i, \bm{\ell}):=\min_{\ell\in \bm{\ell}}\bigl(|x_i-\ell|+e(\ell)\bigr)$. If there is more than one facility that minimizes the cost, we use the following tie-breaking rule:
Let \(\tieset\) be the set of facility locations with equal cost. We choose the facility with the smallest entrance fee. If there are multiple facilities with equal smallest entrance fees, we select the rightmost one.
\footnote{In fact, our results hold under any tie-breaking rule that selects a location from \(\tieset\) in a deterministic and consistent way.}

We consider two classical objectives: the \emph{utilitarian objective} and \emph{egalitarian objective}.
For the utilitarian objective, we want to minimize the total cost of all agents that we denote by:
$TC(\mathbf{x}, \bm{\ell})=\sum_{i=1}^{n}\mathrm{cost}(x_i, \bm{\ell}).$
For the egalitarian objective, we want to minimize the maximum cost of all agents, which is denoted by $MC(\mathbf{x}, \bm{\ell})=max_{i\in N} \bigl(\mathrm{cost}(x_i, \bm{\ell})\bigr)$. 

\begin{remark}\label{remark}
    To keep our exposition general, we assume that there is an oracle that computes $e(x)$ for a given \(x\in \mathbb{R}\), and for any integers \(a\) and \(b\), finds the minimum of $a\cdot e(x)+bx$ in a given interval.
\end{remark}

In the setting of mechanism design, each agent \(i\)'s location \(x_i\) is private. An agent reports a location \(x_i'\) that may differ from her true position. 
A \emph{deterministic mechanism} is a function  \(f(\cdot ,\cdot): \mathcal E \times \mathbb{R}^n\rightarrow \mathbb{R}^m\), where $\mathcal E$ is the set of all entrance fee functions $e: \mathbb R \rightarrow \mathbb R_{\ge 0}$. 
A \emph{randomized mechanism} is a function \(f(\cdot ,\cdot): \mathcal E \times \mathbb{R}^n\rightarrow \Delta(\mathbb{R}^m)\), where $\mathcal E$ is the set of all entrance fee functions $e: \mathbb R \rightarrow \mathbb R_{\ge 0}$ and \(\Delta(\mathbb{R}^m)\) is the set of all probability distributions over \(\mathbb{R}^m\). For a given $e$ and a randomized mechanism \(f(\cdot,\cdot)\), the cost of an agent \(i\in N\) is defined as the expected cost of \(i\), i.e., \(\mathrm{cost}(x_i,f(e,\mathbf{x})):=\mathbb{E}_{\bm{\ell}\sim f(e,\mathbf{x})}[\mathrm{cost}(x_i,\bm{\ell})]\). For any agent position profile \(\mathbf{x}\in \mathbb{R}^n\) the total cost is defined as \(TC(\mathbf{x},f(e,\mathbf{x}))=\mathbb{E}_{\ell\sim f(e,\mathbf{x})}[TC(\mathbf{x},\ell)]\), 
and the maximum cost is defined as \(MC(\mathbf{x},f(e,\mathbf{x}))=\mathbb{E}_{\ell\sim f(e,\mathbf{x})}[MC(\mathbf{x},\ell)]\).

Not only do we consider general mechanisms \(f(\cdot,\cdot)\) taking the entrance fee function as a part of the input, but we also consider mechanisms \(f(e,\cdot)\) for a given entrance fee function $e$.  For the entrance fee function \(e(\cdot)=0\), mechanism \(f(e,\cdot)\) is for the classical model without an entrance fee.
Let \(\alpha\ge 1\) and \(\mathcal{E}(\alpha):=\{e\in \mathcal{E}|r_e=\alpha\}\). 
We may also consider mechanisms under the constraint that the the max-min ratio of the entrance fee function is fixed at a given \(\alpha\), i.e., 
\(f(\cdot ,\cdot): \mathcal{E}(\alpha) \times \mathbb{R}^n\rightarrow \mathbb{R}^m\) or \(f(\cdot ,\cdot): \mathcal{E}(\alpha) \times \mathbb{R}^n\rightarrow \Delta(\mathbb{R}^m)\).

The number of facilities \(m\) and the entrance fee function \(e\) are publicly known.
Given that an agent might misreport her location to decrease her cost, it is necessary to design \emph{(group) strategyproof} mechanisms. Formally, a mechanism \(f(\cdot, \cdot)\) is \emph{group strategyproof} if for any entrance fee function $e$, any profile \(\mathbf{x}\), any coalition \(S\subseteq N\) with any sub-position profile \(\mathbf{x}_S^{\prime}\in \mathbb{R}^{|S|}\),
there exists an agent \(i \in S\) such that \(\mathrm{cost}(x_i, f(e,\mathbf{x}))\le \mathrm{cost}(x_i, f(e,\mathbf{x}'))\), where $ \mathbf{x}'= (\mathbf{x}_S^{\prime}, \mathbf{x}_{-S})$ is obtained from $\mathbf{x}$ by replacing the position profile of all agents in $S$ with $\mathbf{x}_S^{\prime}$. Additionally, if $S$ only contains one agent, we say \(f(\cdot,\cdot)\) is  \emph{strategyproof}.

(Group) strategyproof mechanisms may not be able to achieve the optimal value for one of the two objectives. Thus we use \emph{approximation ratio} to evaluate the performance of the mechanism.
For an entrance fee function $e$ and position profile \(\mathbf{x}\in \mathbb{R}^n\),
let \(OPT_{tc}(e,\mathbf{x})\)
and \(OPT_{mc}(e,\mathbf{x})\)
be the optimal total cost and maximum cost
of the optimization problems, respectively.
For an entrance fee function $e$ and position profile \(\mathbf{x}\), the approximation ratio of  \(f(e,\mathbf{x})\) is defined as
\begin{align*}
    \gamma(f(e,\mathbf{x})):=\frac{TC(\mathbf{x},f(e,\mathbf{x}))}{OPT_{tc}(e,\mathbf{x})},
\end{align*}
the approximation ratio of  \(f(e,\cdot)\) is defined as  \(\gamma(f(e,\cdot)):=\sup_{\mathbf{x}}\gamma(f(e,\mathbf{x}))\), and the approximation ratio of \(f(\cdot,\cdot)\) is defined as \(\gamma(f(\cdot,\cdot)):=\sup_{e\in \mathcal{E}} \gamma(f(e,\cdot))\).
The approximation ratio for the maximum cost is defined in the same way by replacing $\frac{TC(\mathbf{x},f(e,\mathbf{x}))}{OPT_{tc}(e,\mathbf{x})}$ with $\frac{MC(\mathbf{x},f(e,\mathbf{x}))}{OPT_{mc}(e,\mathbf{x})}$.

\section{Structural Properties} \label{sect_stct_prop}
In this section we present some useful insights regarding the structure of agents' optimal location and the solution of the optimization problems. 
As we shall see in the later sections, these properties are useful in designing and analyzing strategyproof mechanisms.
Additionally, we show that the optimization problems can be solved in polynomial time.

\subsection{Monotone and Local Properties}
Given an entrance fee function \(e\!:\!\mathbb{R}\!\rightarrow\! \mathbb{R}_{\ge 0}\) and an agent's position $x \in \mathbb R$, let $x^*$ be the facility location that minimizes the cost of the agent, i.e., \(x^* := \argmin_{\ell\in \mathbb{R}} (|x-\ell|+e(\ell))\).
If multiple locations minimize the cost, we use the tie-breaking rule in \Cref{sect_model}.
We call \(x^*\) the \emph{optimal location for} \(x\), and \(C_i:= \mathrm{cost}(x_i, x_i^*)\) the \emph{optimal cost} for the agent \(i\).

To find the optimal location \(x^*\) for $x$, the definition of \(x^*\) asks us to search the whole space of reals. Observing that if the facility is placed at \(x\), we know the cost for the agent at \(x\) is \(e(x)\). Thus the optimal cost of \(x\) is at most \(e(x)\). Then the global search for \(x^*\) can be reduced to a local search in the neighborhood of \(x\): \(x^* = \argmin_{\ell: |\ell-x|\le e(x)} |\ell-x|+ e(\ell).\)
Then by \Cref{remark}, \(x^*\) can be obtained in a constant time. Next, we derive an important property of \(x^*\).

\begin{lemma} [\textbf{Monotonicity}] \label{lm_xs_monotone}
    For any \(x_i,x_j \in \mathbb{R}\), let \(x_i^*\) and \(x_j^*\) be the optimal location for \(x_i\) and \(x_j\), respectively. Then we have \(x_i^* \leq x_j^*\) if and only if \(x_i\le x_j\).
\end{lemma}
\begin{proof}
    By the definition of \(x_i^*\) and \(x_j^*\) we have
    \begin{align}
        \mathrm{cost}(x_i,x_i^*) - \mathrm{cost}(x_i,x_j^*) & = |x_i-x_i^*|-|x_i-x_j^*|+ e(x_i^*)-e(x_j^*) \le 0, \label{bound_6} \\
        \mathrm{cost}(x_j,x_j^*) - \mathrm{cost}(x_j,x_i^*) & = |x_j-x_j^*|-|x_j-x_i^*|+ e(x_j^*)-e(x_i^*) \le 0. \label{bound_5}
    \end{align}
\noindent
By adding \Cref{bound_6,bound_5}, we get
    \begin{align}
        |x_i-x_i^*|+|x_j-x_j^*| \le |x_i-x_j^*|+|x_j-x_i^*|. \label{bound_7}
    \end{align}
\noindent
Assume for contradiction that \(x_i < x_j\) but \(x_i^* > x_j^*\).  If  \([x_i,x_j]\cap [x_j^*,x_i^*]\ne \emptyset\), then \Cref{bound_7} does not hold.
If  \([x_i,x_j]\cap [x_j^*,x_i^*]=\emptyset\), then \Cref{bound_7} turns into equality. By \Cref{bound_6,bound_5}, we get  \(\mathrm{cost}(x_j,x_j^*) = \mathrm{cost}(x_j,x_i^*)\) and \(\mathrm{cost}(x_i,x_i^*) = \mathrm{cost}(x_i,x_j^*)\).  Then by the tie-breaking rule, we have $x_i^*=x_j^*$. This is a contradiction.

\end{proof}

\noindent
Let  $I[x]$ be the closed interval between \(x\) and \(x^*\). We next depict the relationship between different $I[x]$.

\begin{lemma}\label{lm_Ix_relate}
    For any \(x_i,x_j \in \mathbb{R}\), either \(I[x_i] \cap I[x_j] = \emptyset\) or \(x_i^*=x_j^*\).
\end{lemma}
\begin{proof}
    Suppose w.l.o.g. that \(x_i <  x_j\) and $x_i\le x_i^*$.
    Then \(x_i^* \leq x_j^*\) by \Cref{lm_xs_monotone}.
    Suppose for contradiction that \(I[x_i] \cap I[x_j] \ne \emptyset\) and \(x_i^*\ne x_j^*\).
    Then we have $x_i < x_j \leq x_i^* < x_j^*$.  By
    \begin{align*}
        \mathrm{cost}(x_i, x_i^*)
        = x_j-x_i + \mathrm{cost}(x_j, x_i^*) 
        \ge x_j-x_i + \mathrm{cost}(x_j, x_j^*)
        = \mathrm{cost} (x_i, x_j^*),
    \end{align*}
    we have $\mathrm{cost}(x_i, x_i^*)\geq \mathrm{cost} (x_i, x_j^*)$. By the definition of \(x_i^*\), we know that $\mathrm{cost}(x_i, x_i^*) = \mathrm{cost} (x_i, x_j^*)$. Then by the tie-breaking rule, we have $x_i^*=x_j^*$, a contradiction.

\end{proof}

\noindent
By Lemmas \ref{lm_xs_monotone} and \ref{lm_Ix_relate}, we can prove the following property for three positions.

\begin{lemma}\label{lm_3_monotone}
If \(x_i \le x_j \le x_k\), then \(\mathrm{cost}(x_i,x_j^*)  \le \mathrm{cost}(x_i,x_k^*)\) and  \(\mathrm{cost}(x_k,x_j^*)\le \mathrm{cost}(x_k,x_i^*)\).
\end{lemma}
\begin{proof}
    By \Cref{lm_xs_monotone}, we get $x_i^*\leq x_j^* \leq x_k^*$. If $x_i^*=x_j^*$ or $x_j^*=x_k^*$, then this lemma is clearly true.
    Thus, we assume that $x_i^* < x_j^* < x_k^*$.
    Then by Lemma \ref{lm_Ix_relate}, we have that \(I[x_i], I[x_j]\) and \(I[x_k]\) are pairwise disjoint. Thus it holds that \(x_i< x_j < x_k^*\).
    Then by \(\mathrm{cost}(x_j,x_k^*)\ge \mathrm{cost}(x_j,x_j^*)\) and the triangle inequality, we have
    \begin{align*}
        \mathrm{cost}(x_i,x_k^*)
        &= |x_i-x_j|+|x_j-x_k^*|+e(x_k^*)\\
        &\ge |x_i-x_j|+|x_j-x_j^*|+e(x_j^*)\\
        &\ge |x_i-x_j^*|+e(x_j^*) =\mathrm{cost}(x_i,x_j^*).
    \end{align*}
    Similarly, we can prove \(\mathrm{cost}(x_k,x_j^*)\le \mathrm{cost}(x_k,x_i^*)\).

\end{proof}

\noindent
Next, we capture the property that one location is preferred by each agent over another location: we say location \(\ell_1\) \emph{dominates} location \(\ell_2\) if \(\mathrm{cost}(x_i,\ell_1)\!\le \mathrm{cost}(x_i,\ell_2)\) for all \(i\in N\).

\begin{lemma} \label{lm_xs_dominate_Ix}
    For any position \(x\in \mathbb{R}\) we have
    \(x^*\) dominates all locations in \(I[x]\).
\end{lemma}
\begin{proof}
    Assume w.l.o.g. that \(x \le x^*\). Let \(\ell\in I[x]\) and \(i\in N\).
    For each agent $i$, there are three cases. 

    \noindent
   \textbf{Case 1}: $x_i\leq \ell$. Note that $e(\ell)\geq x^*-\ell + e(x^*)$ as  $\mathrm{cost} (x, \ell)\geq \mathrm{cost} (x, x^*)$.
    Therefore, we have
      \begin{align*}
      \mathrm{cost}(x_i, \ell) =\ell-x_i+ e(\ell)
      \geq (\ell-x_i)+(x^*-\ell)+ e(x^*) 
      =\mathrm{cost} (x_i, x^*).
      \end{align*}
     \noindent
    \textbf{Case 2}: $\ell < x_i < x^*$.   We have the following calculations:
   \begin{align*}
    \mathrm{cost}(x_i, \ell) 
    = x_i-\ell + e(\ell)
    \geq x_i -\ell+ x^*-\ell+e(x^*)
    \geq x^*-x_i+e(x^*)
    =\mathrm{cost} (x_i, x^*).
   \end{align*}
         \noindent
    \textbf{Case 3}: $x^*\leq x_i$. Similarly we can prove  $\mathrm{cost}(x_i, \ell)\leq \mathrm{cost} (x_i, x^*)$.
\end{proof}

\subsection{Solving the Optimization Problems}
In this subsection, we delve into the algorithmic perspective of facility location problems with entrance fees, demonstrating that the optimization problems can be efficiently solved in polynomial time.
\footnote{While the main focus of this paper is on mechanism design, our algorithmic results may hold independent significance within the realm of exactly solving one-dimensional facility location problems (e.g., \cite{love1976adynamic,amaral08anexact,brimberg2001capacitated,chen2014new,chen2015efficient}).}

By \Cref{lm_xs_monotone},  we have \(x_1^* \le x_n^*\). For \(m=1\), let \(\ell_{tc}\) be the location that minimizes the total cost, i.e.,
\(\ell_{tc}= \argmin_{\ell\in \mathbb{R}}TC(\mathbf{x},\ell)\). If there are multiple solutions, then we use the tie-breaking rule in \Cref{sect_model}.
Define \(\ell_{mc}\) as the location that minimizes the maximum cost similarly.

\begin{lemma} \label{lm_ltc_bounded}
    \(\ell_{tc}\) and \( \ell_{mc} \textrm{ are within the interval } [x_1^*,x_n^*]\).
\end{lemma}

\noindent
We want to design effective algorithms that find the optimal facility location profile for either objective when the position profile of agents is given.

\begin{lemma}\label{lm_1f_alg}
    When only one facility exists, the two optimization problems can be solved in \(O(n)\) time.
\end{lemma}

\noindent
When \(m>1\), we have the following lemma that describes the structure of a solution.
\begin{lemma} \label{lm_continuous}
    For a given facility location profile, if two agents select the same facility, then any agent located between these two agents also selects the same facility.
\end{lemma}
\noindent
The proofs of \Cref{lm_ltc_bounded,lm_1f_alg,lm_continuous} can be found in the appendix. Equipped with these lemmas,  we show in next proposition that both optimization problem can be solved in polynomial time.

\begin{proposition}
    A solution minimizing the total cost or maximum cost can be found in \(O(n^3m)\) time.
\end{proposition}
\begin{proof}
    We show that both two optimization problems can be solved by using dynamic programming.
    Let \(N_j\) be the set of customers of facility \(j\), \(\ell_j\) be the location of facility \(j\), then a solution is denoted by \(\{(\ell_j,N_j)\}_{j=1}^{m}\) where \(\bigcup_{j=1}^{n} N_j = N\) and the sets \(N_j\) are pairwise disjoint.
    Recall that \(x_1\le \dots \le x_n\), we say a subset of agents \(N_j\subseteq N\) is \emph{continuous} if  \(N_i=\{x_{j_1},x_{j_1+1}, \dots, x_{j_2}\}\) for some \(1\le j_1 \le j_2 \le n\).
    By Lemma 8, we have that \(N_j\) is continuous. Denote by \([i,j]\) the continuous set of agents \(\{i, i+1, \dots, j\}\).

    First, we consider the problem of minimizing the total cost.
    Denote by \(OPT_{tc}(i,j,k)\) the minimum total cost when agents \([1,j]\) are divided into \(k\) partitions, and \(i\le j\) are in the \(k\)th partition. Thus the optimal total cost to our problem is \(OPT_{tc}(n,n,m)\).
    Denote by \(v_{tc}(i,j)\) the optimal total cost incurred by agents in \([i,j]\) when they are served by a single facility. By Lemma \ref{lm_1f_alg}, \(v_{tc}(i,j)\) can be computed in \(O(n)\) time.

    We will make use of an \(n\times n \times m\) array \(M\), whose entries are initially set to empty.
    We invoke Algorithm \ref{algo_dp_tc} to compute \(OPT_{tc}(n,n,m)\) and recover the solution from the values stored in \(M\).

    \begin{algorithm}[t] 
        \caption{\(OPT_{tc}(i,j,k)\)} 
        \begin{algorithmic}[1] 
        \IF {\(i=0\) or \(j=0\) or \(k=0\)}
        \STATE Return {\(0\)}.
        \ELSIF{\(M(i,j,k)\) not empty}
        \STATE Return {\(M(i,j,k)\)}.
        \ELSE
        \STATE \(M(i,j,k)\gets \min\bigl(OPT_{tc}(i-1,j,k), OPT_{tc}(i-1,i-1,k-1)+v_{tc}(i,j)\bigr)\) \\
                Return {\(M(i,j,k)\)}
        \ENDIF
        \end{algorithmic}
    \end{algorithm} \label{algo_dp_tc}

    We argue for correctness. It is easy to see that the running-time bound of Algorithm \ref{algo_dp_tc} is \(O(mn^3)\). The base case in line \(1\) of Algorithm \ref{algo_dp_tc} is clear. The total cost is zero if there are no agents or no facilities. Suppose we want to compute \(OPT_{tc}(i,j,k)\) and we have already computed \(OPT_{tc}(i',j',k')\) where \(i' < i\) or \(j' < j\) or \(k' < k\). We distinguish between two cases depending on whether or not agent \(i-1\) is still in the \(k\)th partition.
    If agent \(i-1\) is in the \(k\)th partition, \(OPT_{tc}(i,j,k)=OPT_{tc}(i-1,j,k)\). Otherwise we know that \(i-1\) is in right boundary the \((k-1)\)th partition and agents in \([i,j]\) are served by a single facility, thus \(OPT_{tc}(i,j,k)=OPT_{tc}(i-1,i-1,k-1)+v_{tc}(i,j)\).
    Actually, we derive the following transition function:
    \begin{align*}
        OPT_{tc}(i,j,k)&=\min\bigl(OPT_{tc}(i-1,j,k), OPT_{tc}(i-1, i-1,k-1)+v_{tc}(i,j)\bigr).
    \end{align*}
    It is easy to see that Algorithm~\ref{algo_dp_tc} can be implemented in \(O(mn^3)\) time by using memorization.

    Let \(OPT_{mc}(i,j,k)\) be the minimum maximum cost of agents \([1,j]\) under the conditions that agents \([1,j]\) are divided into \(k\) partitions, and agent \(i\le j\) is in the \(k\)th partition. Let \(v_{mc}(i,j)\) be the optimal maximum cost of agents \([i,j]\) when they are served by a single facility.
    In a slightly different way, we derive the transition function for the maximum cost as follows:
    \begin{align*}
        OPT_{mc}(i,j,k)=\min\Bigl(OPT_{mc}(i-1,j,k),
        \min\bigl(OPT_{mc}(i-1,i-1,k-1), v_{mc}(i,j)\bigr)\Bigr).
    \end{align*}
    The algorithm to compute  \(OPT_{mc}(i,j,k)\) is similar to Algorithm \ref{algo_dp_tc}.
\end{proof}

\section{Utilitarian Version with One Facility} \label{sect_1f_tc}
In this section, we study the one-facility game with the objective of minimizing the total cost.
We will first investigate upper and lower bounds for deterministic mechanisms and then focus on randomized mechanisms.

\subsection{Deterministic Mechanisms}

Recall that it is assumed that \(x_1\le \dots \le x_n\), and for any \(i\in N\), we denote by \(m_i(\cdot,\cdot)\) the mechanism that outputs the optimal location $x_i^*$ for agent whose location is \(x_i\).

\begin{theorem} 
The mechanism \(m_i(\cdot,\cdot)\) is group strategyproof. 
\end{theorem}
\begin{proof}
    Obviously, any agent at \(x_i\) has no incentive to misreport. To modify the output of \(m_i(\cdot,\cdot)\), at least one agent located in \([x_1,x_i)\) has to misreport a location strictly greater than \(x_i\) or at least one agent in \((x_i,x_n]\) has to misreport a location strictly smaller than \(x_i\). We will prove that the agent will not benefit from misreporting in either case. 

    We only consider the case that agent \(j\) with \(x_j< x_i\) misreports \(x_j'> x_i\). The other case is proved similarly.
    Let \(x_i'\) be the \(i\)th order statistic in \((x_j',\mathbf{x}_{-j})\). 
	Then we have $x_j< x_i \leq x_i'$.
    Let \(\mathbf{x}'=(x_j',\mathbf{x}_{-j})\) be the misreported profile and $x_i^{\prime *} =m_i(e,\mathbf{x}')$ be the output of the misreported profile.
    By Lemma \ref{lm_3_monotone}, we get \(\mathrm{cost}(x_j,x_i^*)\le \mathrm{cost}(x_j,x_i^{\prime *})\). Thus agent \(j\) has no incentive to misreport. 
\end{proof}

Let \(m_{med}(\cdot, \cdot)\) be the mechanism that always outputs the optimal location $x_{med}^*$ for the median agent \(med=\lceil n/2 \rceil\).
Let  \(\mathbf{x}\in \mathbb{R}^n\), \(\ell\in \mathbb{R}\) and \(e_{\ell}\in \mathbb{R}_{\ge 0}\).
To ease our analysis of the approximation ratio of \(m_{med}(\cdot, \cdot)\) for the total cost,  we introduce the notion of \emph{virtual total cost} $VT(\mathbf{x},\ell,e_{\ell}):=\sum_{i\in N}|x_i-\ell |+n\cdot e_{\ell}$, which is the total cost of \(\mathbf{x}\) when a \emph{virtual facility} is located at \(\ell\) with entrance fee \(e_{\ell}\).
It is possible \(e_{\ell}\ne e(\ell)\). When \(e_{\ell}=e(\ell)\), we have  \(VT(\mathbf{x},  \ell, e_{\ell})= TC(\mathbf{x},\ell)\).
The \emph{virtual cost} of agent $i$ with respect to the virtual facility is \(\mathrm{vcost}(x_i, \ell, e_{\ell}):=|x_i-\ell|+e_{\ell}\). Clearly,  \(\mathrm{vcost}(x_i,\ell, e(\ell))=\mathrm{cost}(x_i,\ell)\) when \(e_{\ell}=e(\ell)\).

\begin{observation} \label{obs_VT}
    Let \(\mathbf{x}\in \mathbb{R}^n\), \(\ell_1,\ell_2\in \mathbb{R}\) and \(e_1, e_2 \in \mathbb R_{\ge 0}\) such that $e_1< e_2$.
    If  $\ell_2\leq \ell_1\leq x_{med}$ or $x_{med}\leq \ell_1\leq \ell_2$, then
\(VT(\mathbf{x},\ell_1,e_1)< VT(\mathbf{x},\ell_1,e_2) \le VT(\mathbf{x},\ell_2,e_2)\).
\end{observation}

\noindent
Next, we show that the total cost approximation ratio of  \(m_{med}(\cdot,\cdot)\) can be bounded by a constant.

\begin{theorem} \label{thm_1f_tc_apx}
For each entrance fee function $e$ with max-min ratio \(r_e\), the approximation ratio of \(m_{med}(e,\cdot)\) for the total cost is at most \(3-\frac{4}{r_e+1}\).
Hence, the approximation ratio of $m_{med}(\cdot,\cdot)$ \  for the total cost is at most \ \(3\).
\end{theorem}

\begin{proof}
    Assume w.l.o.g. that \(x_{med} \le x_{med}^*\). Let $C_{med}$ be the optimal cost of the median agent.
    For any agent location \(x_i\), if \(x_i\!\le\! x_{med}\), we have \(\mathrm{cost}(x_i,x_{med}^*)\!=\!\mathrm{vcost}(x_i,x_{med},C_{med})\). If \(x_i\!>\! x_{med}\), by Lemma \ref{lm_xs_dominate_Ix} we have \(\mathrm{cost}(x_i,x_{med}^*)\!\le\! \mathrm{vcost}(x_i,x_{med},C_{med})\). Thus
    \begin{align}
        TC(\mathbf{x}, x_{med}^*) \le VT(\mathbf{x},x_{med},C_{med}). \label{bound_m_med}
    \end{align}
    Next we prove \(e(\ell_{tc})\le C_{med}\). Suppose for contradiction that \(e(\ell_{tc})> C_{med}\). Then by \Cref{obs_VT} and \Cref{bound_m_med}, we have \(TC(\mathbf{x},\ell_{tc}) > VT(\mathbf{x},x_{med}, C_{med}) \ge TC(\mathbf{x},x_{med}^*)\).
    This is a contradiction.
    Hence, there is a location \(\ell_{tc}'\) between \(x_{med}\) and \(\ell_{tc}\) such that \(\mathrm{vcost}(x_{med}, \ell_{tc}', e(\ell_{tc}))= C_{med}\). Then we have \(|x_{med}-\ell_{tc}'|=C_{med}-e(\ell_{tc})\) . Let \(\Delta= |C_{med}-e(\ell_{tc})|\).
    By Observation \ref{obs_VT}, we have \(TC(\mathbf{x}, \ell_{tc})\ge VT(\mathbf{x},\ell_{tc}',e(\ell_{tc}))\).
    Together with \Cref{bound_m_med}, we have
    \begin{align}
        \gamma(m_{med}(e,\mathbf{x})) &\le \frac{VT(\mathbf{x},x_{med},C_{med})}{VT(\mathbf{x},\ell_{tc}',e(\ell_{tc}))}. \label{bound_apx}
    \end{align}
    We will identify a position profile \(\mathbf{x}'\in \mathbb{R}^n\) for each \(\mathbf{x}\) such that
    the bound given by \Cref{bound_apx} for \(\mathbf{x}\) is not greater than that for \(\mathbf{x}'\).
    Let \(n_1=|\{i\in N|x_i\le x_{med} < \ell_{tc}' \text{ or } x_i \ge x_{med}\ge \ell_{tc}'\}|\) and \(n_2=n-n_1\). Thus \(n_1\ge n_2\).
    Let \(\mathbf{x}'\) be the profile such that \(n_1\) agents are in \(x_{med}\) and \(n_2\) agents are in \(\ell_{tc}'\).
    By \Cref{bound_apx},
    \begin{align}
        \gamma(m_{med}(e,\mathbf{x}))
        \le \frac{VT(\mathbf{x}',x_{med}, C_{med})}{VT(\mathbf{x}',\ell_{tc}',e(\ell_{tc}))}
        = \frac{(C_{med}+\Delta)\cdot \frac{n_2}{n_1}+C_{med}}{e(\ell_{tc})\cdot \frac{n_2}{n_1}+C_{med}}. \label{tc_apx_n1n2}
      \end{align}
      Since \(C_{med}+\Delta \ge e(\ell_{tc})\) and \(n_2/n_1 \le 1\), we know that \Cref{tc_apx_n1n2} is maximized when \(n_1 = n_2\). Then since \(e_{\mathrm{min}}\le e(\ell_{tc}) \le C_{med} \le e_{\mathrm{max}}\), we have
      \(\gamma(m_{med}(e,\mathbf{x}))\le \frac{2C_{med}+\Delta}{C_{med}+e(\ell_{tc})} \le 3 - \frac{4}{r_e+1}.\)
        This proves the first claim of the theorem. The second claim follows immediately.

\end{proof}

The example in Figure \ref{fig_tc_tight} shows that the approximation ratio of \(3 - \frac{4}{r_e+1}\) in Theorem \ref{thm_1f_tc_apx} is tight for our mechanism, for all even \(n\). Let \(L=e_{\mathrm{max}}-e_{\mathrm{min}}+\epsilon\), where \(e_{\mathrm{max}}\ge e_{\mathrm{min}}\ge 0\) and \(\epsilon\) is a small positive.
We consider the following entrance fee function: \(e(L) = e_{\mathrm{min}}\) and \(e(\ell) = e_{\mathrm{max}}\) for any \(\ell\ne L\).
Consider the position profile \(\mathbf{x}=(0,L)\) with two agents, the optimal location is \(e_{\mathrm{min}}\) with the total cost \(2e_{\mathrm{min}}+L\),
and the facility location returned by \(m_{med}(\cdot,\cdot)\) is $0$, with the total cost \(2e_{\mathrm{max}}+L\).

\begin{figure} [t]
    \centering
    \includegraphics[width=\myscale\linewidth]{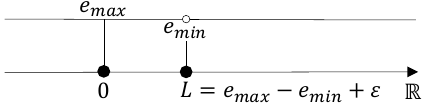}
    \caption{A tight example for the approximation ratio of mechanism $m_{med}(\cdot,\cdot)$. The black dots are agents. The upper horizontal line represents the entrance fee function \(e\). The height of the vertical line segment represents the entrance fee of the location.} 
    \label{fig_tc_tight}
\end{figure} 

Next, we establish a lower bound of \(3\) for the total cost approximation ratio against all deterministic strategyproof mechanisms. Thus we can assert that no deterministic strategyproof mechanism can do better than \(m_{med}(\cdot, \cdot)\).

\begin{theorem} \label{thm_1f_tc_d_lb_3re}
    For any given \(\alpha\ge 1\), there exists an entrance fee function \(e\) with max-min ratio \(r_e=\alpha\) such that no deterministic strategyproof mechanism \(f(e,\cdot):\mathbb{R}^n\rightarrow \mathbb{R}\) can achieve an approximation ratio less than \(3 - \frac{16}{r_{e} + 5 + \sqrt{r_{e}^{2} + 10r_{e} - 7}}\) for the total cost.
    Hence, no deterministic strategyproof mechanism \(f(\cdot, \cdot):\mathcal{E} \times \mathbb R^n \rightarrow \mathbb R\) can achieve a total cost approximation ratio less than $3$.
\end{theorem}

\begin{proof}
    Let \(d>0\) and \(D = d + 1 + (2d + 1)^{-1}\).
    Consider the following entrance fee function:
    \begin{equation*}
            e(\ell) = \left\{
            \begin{aligned}
                &d, &&\text{if } \ell \in \{-1, 1\} \\
                &D, &&\text{otherwise}
            \end{aligned}
            \right.,
    \end{equation*}
    with $r_{e} = \frac{D}{d}$.
    Let $f(e,\cdot)$ be a deterministic strategyproof mechanism.
    Let $\mathbf{x}_{1} = (-1, \epsilon)$, $\mathbf{x}_{2} = (-\epsilon, 1)$ and $\mathbf{x}_{3} = (-1, 1)$ be three agent position profiles, where $\epsilon$ is a small positive.
    We illustrate this example in Figure \ref{fig_tc_lb_d3}.

\begin{figure} [t]
      \centering
      \includegraphics[width=\myscale\linewidth]{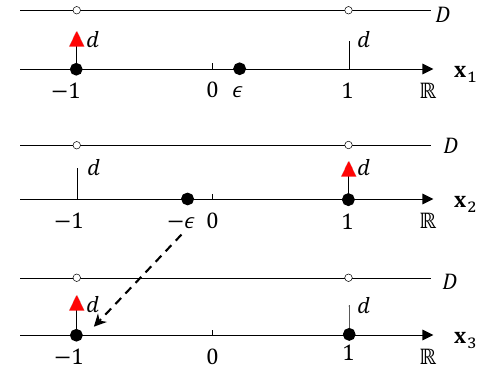}
      \caption{Definitions of \(e(\cdot),\mathbf{x}_1,\mathbf{x}_2 \textrm{ and } \mathbf{x}_3\). Red triangles are facilities, and the dashed line denotes the deviation of the agent.} 
      \label{fig_tc_lb_d3}
\end{figure}

    First, we consider profile $\mathbf{x}_{1}$, we have
    \begin{equation*}
          TC(\mathbf{x}_{1}, \ell) = \left\{
          \begin{aligned}
            &2d + 1 + \epsilon, &&\text{if } \ell = -1 \\
            &2d + 3 - \epsilon, &&\text{if } \ell = 1 \\
            &|\ell + 1| + |\ell - \epsilon| + 2D, &&\text{otherwise} \\
          \end{aligned}
          \right..
    \end{equation*}
    When \(\ell\ne \pm 1\), we have \(TC(\mathbf{x}_1,\ell)= |\ell + 1| + |\ell - \epsilon| + 2D\ge 2d+3\). Thus \(OPT_{tc}(e,\mathbf{x}_1)=2d + 1 + \epsilon\).
    Assume, for contradiction, that
    \begin{align} \label{tc_asp_apx}
        \gamma(f(e, \cdot)) < \frac{2d + 3 - \epsilon}{2d + 1 + \epsilon}.
    \end{align}
    Then we have \(TC(\mathbf{x}_1,f(e,\mathbf{x}_1))\!\le\! \gamma(f(e, \cdot))\!\cdot\!OPT_{tc}(e,\mathbf{x}_1)\!<\!2d+3-\epsilon\).
    Thus \(TC(\mathbf{x}_1,f(e,\mathbf{x}_1)) \textrm{ must be }2d+1+\epsilon\). Then we have $f(e,\mathbf{x}_{1}) = -1$ and $\mathrm{cost}(\epsilon,f(e,\mathbf{x}_{1})) = d + 1 + \epsilon$.
    By the symmetry between profiles $\mathbf{x}_{1}$ and $\mathbf{x}_{2}$, we have $f(e, \mathbf{x}_{2}) = 1$ and $\mathrm{cost}(-\epsilon, f(e,\mathbf{x}_{2}))=d + 1 + \epsilon$.

    Now consider profile $\mathbf{x}_{3}$.
    We have
    \begin{equation*}
            TC(\mathbf{x}_{3}, \ell) = \left\{
            \begin{aligned}
                &2d + 2, &&\text{if } \ell = \pm 1 \\
                &|\ell + 1| + |\ell - 1| + 2D, &&\text{otherwise} \\
            \end{aligned}
            \right..
    \end{equation*}
    Thus, the optimal total cost for \(\mathbf{x}_3\) is \(OPT_{tc}(e,\mathbf{x}_3)=2d+2\) when the facility is located at \(-1\) or \(1\).
    By our assumption \Cref{tc_asp_apx} and the definition of \(D\), we have
    \[\gamma(f(e, x_{3})) < \frac{2d + 3}{2d + 1}=\frac{2D+2}{2d+2}.\]
    Since\(|\ell + 1| + |\ell - 1| + 2D \ge 2D+2\), we know that \(TC(\mathbf{x}_3,f(e,\mathbf{x}_3))\) must be \(2d+2\).
    Thus $f(e, \mathbf{x}_{3}) = \pm 1.$
    If $f(e,\mathbf{x}_{3}) = -1$, then the agent at $-\epsilon$ in $\mathbf{x}_{2}$ can deviate to $-1$ and decrease her cost from $d + 1 + \epsilon$ to $d + 1 - \epsilon$, a contradiction.
    If $f(e,\mathbf{x}_{3}) = 1$, then the agent at $\epsilon$ in $\mathbf{x}_{1}$ can deviate to $1$ and decrease her cost from $d + 1 + \epsilon$ to $d + 1 - \epsilon$, a contradiction.

    Due to the arbitrariness of $\epsilon > 0$,
    we have \(\gamma(f(e, \cdot))\ge \frac{2d + 3}{2d + 1}\).
    Observing that $r_{e} = 1 + \frac{2d + 2}{2d^2 + d}$, we can get the desired result after simple computation.
\end{proof}

\subsection{Randomized Mechanisms}

For any input \(e\) and \(\mathbf{x}\), our randomized mechanism, denoted by \(r(\cdot, \cdot)\), outputs each \(x_i^*\) with probability \(\frac{1}{n}\) for \(i\in N\). Recall that for any \(n\), the approximation ratio of the deterministic mechanism \(m_{med}(\cdot,\cdot)\) is exactly \(3\). Next we show that the randomized mechanism \(r(\cdot, \cdot)\) can achieve a better approximation ratio of \(3-\frac{2}{n}\).

\begin{theorem} \label{thm_1f_tc_r_ub}
    The mechanism \(r(\cdot,\cdot)\) is strategyproof and achieves an approximation ratio of \(3-\frac{2}{n}\) for the total cost, where \(n\) is the number of agents.
    The approximation ratio is tight for this mechanism.
\end{theorem}
\begin{proof}
    We first prove the strategyproofness of the mechanism \(r(\cdot,\cdot)\).
    The expected cost for agent \(i\) is \(\mathrm{cost}(x_i, r(e,\mathbf{x}))=\frac{1}{n}\sum_{j\in N}\mathrm{cost}(x_i,x_j^*)\).
    Suppose agent \(i\) deviates to \(x_i'\). Let \(x_i'^*\) be the optimal location for \(x_i'\). By the definition we have \(\mathrm{cost}(x_i,x_i^*)\le \mathrm{cost}(x_i,x_i'^*)\). Since the optimal location for agent \(j\ne i\) does not change, we have \(\mathrm{cost}(x_i, r(e,\mathbf{x}))\le \mathrm{cost}(x_i,r(e,\mathbf{x}'))\), where \(\mathbf{x}'=(x_i',\mathbf{x}_{-i})\). Thus, agent $i$ can not reduce her cost by misreporting.

    Now we calculate the approximation ratio of \(r(\cdot, \cdot)\). Given \(e\) and \(\mathbf{x}\), the optimal total cost \(OPT_{tc}\) is
    \begin{align}
        TC(\mathbf{x},\ell_{tc})&=\sum_{i\in N}\mathrm{cost}(x_i,\ell_{tc})
        \ge \sum_{i\in N} C_i. \label{tc_opt}
    \end{align}
    The last inequality is due to \(\mathrm{cost}(x_i,\ell_{tc})\ge \mathrm{cost}(x_i,x_i^*)=C_i\) for any \(i\in N\). The total cost of \(r(\cdot, \cdot)\) is
    \begin{align}
        TC(\mathbf{x}, r(e,\mathbf{x})) \nonumber
        &=\frac{1}{n}\sum_{i\in N}\sum_{j\in N}\mathrm{cost}(x_j, x_i^*) \nonumber 
        =\frac{1}{n}\sum_{i\in N}C_i+\frac{1}{n}\sum_{(i,j)\in N^2:i\ne j}\mathrm{cost}(x_i,x_j^*) \nonumber\\
        &\le \frac{1}{n}OPT_{tc}+\frac{1}{n}\sum_{(i,j)\in N^2:i\ne j}\mathrm{cost}(x_i,x_j^*),  \label{tc_r_1}
    \end{align}
    where the last inequality holds because of \Cref{tc_opt}.
    Then, we only need to bound the second term of \Cref{tc_r_1}.
    For any \(i,j\in N\), by the triangle inequality, we have \(\mathrm{cost}(x_i,x_j^*)\le|x_i-x_j|+|x_j-x_j^*|+e(x_j^*)=|x_i-x_j|+C_j\). Thus,
    \begin{align}
        \sum_{(i,j)\in N^2:i\ne j}\mathrm{cost}(x_i,x_j^*) 
        &=\sum_{k=1}^{n-1}\sum_{|i-j|=k}\mathrm{cost}(x_i,x_j^*)
        =\sum_{k=1}^{n-1}\sum_{i-j=k}(\mathrm{cost}(x_i,x_j^*)+\mathrm{cost}(x_j,x_i^*))\nonumber \\
        &\le \sum_{k=1}^{n-1}\sum_{i-j=k}(C_i+C_j+2|x_i-x_j|) \nonumber\\
        &=\sum_{k=1}^{n-1}\sum_{i-j=k}(C_i+C_j)+2\sum_{k=1}^{n-1}\sum_{i-j=k}|x_i-x_j|. \label{tc_r_2}
    \end{align}
    For the first term of \Cref{tc_r_2}, we have
    \begin{align}
        \sum_{k=1}^{n-1}\sum_{i-j=k}(C_i+C_j) =(n-1)\sum_{i\in N}C_i
        \le (n-1)OPT_{tc}. \label{tc_r_3}
    \end{align}
    Since \(OPT_{tc}\ge \sum_{i\in N}|x_i-x_{med}|\), for the second term of \Cref{tc_r_2}, we have
    \begin{align}
        2 \sum_{k=1}^{n-1}\sum_{i-j=k}|x_i-x_j| &\le 2 \sum_{k=1}^{n-1}\sum_{i\in N} |x_i-x_{med}| 
        \le 2(n-1)OPT_{tc}. \label{tc_r_4}
    \end{align}
    \noindent
    By applying \Cref{tc_r_3,tc_r_4} to \Cref{tc_r_2}, and then applying \Cref{tc_r_2} to \Cref{tc_r_1}, we have \( TC(\mathbf{x},r(e,\mathbf{x})) \le (3-\frac{2}{n})OPT_{tc}\).

\begin{figure} [http]
      \centering
      \includegraphics[width=\myscale\linewidth]{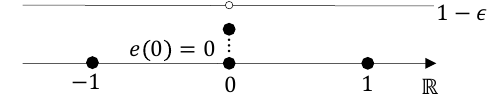}
      \caption{A tight example for the approximation ratio in Theorem \ref{thm_1f_tc_r_ub}. The black dots are agents.} 
      \label{fig_tc_r_tight}
\end{figure}
    To finish our proof, we use the example of \(n\) agents in Figure \ref{fig_tc_r_tight} to show the upper bound of \(3-\frac{2}{n}\) is tight for the mechanism \(r(\cdot,\cdot)\). Consider the agent position profile \(\mathbf{x}=(-1, 0,\dots,0, 1)\) and the entrance fee function \(e(\cdot)\), where \(e(\ell)=0\) when \(\ell = 0\) and \(e(\ell)=1-\epsilon\) when \(\ell\ne 0\). The optimal location is \(0\) with a total cost of \(2\), and the expect total cost of mechanism \(r(\cdot, \cdot)\) is \(\frac{2}{n}\!\cdot\!(n(1-\epsilon)+n)+\frac{n-2}{n}\!\cdot\! 2=6-\frac{4}{n}-2\epsilon\).
\end{proof}

Next, we establish a lower bound for the approximation ratio against all randomized strategyproof mechanisms.

\begin{theorem} \label{thm_1f_tc_r_lb_2}
    There exists an entrance fee function \(e\) with \(r_e=+\infty\) such that no randomized strategyproof mechanism \(f(e, \cdot): \mathbb{R}^n\rightarrow \Delta(\mathbb{R})\) can achieve a total cost approximation ratio less than \(2\). Hence, no randomized strategyproof mechanism \(f(\cdot, \cdot):\mathcal{E} \times \mathbb R^n \rightarrow \Delta(\mathbb R)\) can achieve a total cost approximation ratio less than $2$.
\end{theorem}

\begin{proof}
    Consider the following entrance fee function:
    \begin{equation*}
        e(\ell) =
        \begin{cases}
          \begin{aligned}
              &0, &&\text{if } \ell \in \{ -1, 1\} \\
              &+\infty, &&\text{otherwise} \\
            \end{aligned}
        \end{cases}.
    \end{equation*}
    Let \(f(e,\cdot):\mathbb{R}^n\rightarrow \Delta(\mathbb{R})\) be a randomized strategyproof mechanism.
    To achieve a bounded  total cost approximation ratio, for any profile \(\mathbf{x}\in \mathbb{R}^n\), it holds that \(\mathrm{Pr}[f(e,\mathbf{x})=-1] + \mathrm{Pr}[f(e,\mathbf{x})=1]=1.\)
    Let \(\mathbf{x}_1=(-1,\epsilon)\), \(\mathbf{x}_2=(-\epsilon,1)\) and \(\mathbf{x}_3=(-1,1)\) be three agent position profiles, where \(\epsilon \in (0, 1)\) is a small positive.
    Let \(\mathrm{Pr}[f(e,\mathbf{x}_1)=1]=p\) and \(\mathrm{Pr}[f(e,\mathbf{x}_2)=-1]=p'\).
    We illustrate the example in Figure \ref{fig_tc_lb_r3}.
\begin{figure} [t]
      \centering
      \includegraphics[width=\myscale\linewidth]{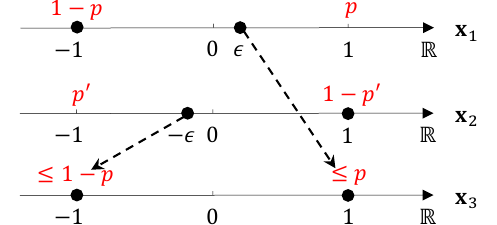}
      \caption{Example for randomized lower bound \(2\) for the total cost. The black dots are agents. The red letters denote the probabilities to locate the facility by the mechanism \(f(e,\cdot)\). The dashed lines represent the deviations of the agents.} 
      \label{fig_tc_lb_r3}
\end{figure}
    \(TC(\mathbf{x}_1,-1)=1+\epsilon\) is the optimal total cost and \(TC(\mathbf{x}_1,1)=3-\epsilon\). Thus, the approximation ratio of \(f\) for \(\mathbf{x}_1\) is
    \begin{equation} \label{tc_bound_1}
        \gamma(f(e,\mathbf{x}_1))=(1-p)\cdot 1+p\cdot \frac{3-\epsilon}{1+\epsilon}.
    \end{equation}

    Now consider the case that the agent at \(\epsilon\) in profile \(\mathbf{x}_1\) deviates to \(1\). Since this agent prefers location \(1\) to location \(-1\), by the strategyproofness of \(f(e,\cdot)\), we must have \(\mathrm{Pr}[f(e,\mathbf{x}_3)=1] \le p\). Then consider the case that the agent at \(-\epsilon\) in \(\mathbf{x}_2\) deviates to \(-1\), again by the strategyproofness of \(f(e,\cdot)\), we must have \(\mathrm{Pr}[f(e,\mathbf{x}_3)=-1] \le p'\). Thus \(p+p' \ge \mathrm{Pr}[f(e,\mathbf{x}_3)=-1]+\mathrm{Pr}[f(e,\mathbf{x}_3)=1]=1.\) W.l.o.g., we assume \(p\ge 1/2\).
    Then by \Cref{tc_bound_1}, we have \(\gamma(f(e,\cdot))\ge \gamma(f(e,\mathbf{x}_1))\ge \frac{2}{1+\epsilon}.\)
    By the arbitrariness of \(\epsilon>0\), we know that the total cost approximation ratio of any randomized strategyproof mechanism is at least \(2\).

\end{proof}

Thus, for \(n=2\), our randomized mechanism \(r(\cdot,\cdot)\) achieves the best possible approximation ratio.
Note that the max-min ratio in the above proof is \(+\infty\). Thus we can safely assume that the probability distribution of the mechanism is discrete. However, when the max-min ratio is arbitrarily given, this assumption does not hold, and we have to deal with continuous distributions.
In the following theorem, we show that we can discretize continuous distributions using the first mean-value theorem for integrals and establish a lower bound against all the entrance fee functions of a given max-min ratio.

\begin{theorem} \label{thm_1f_tc_r_lb_re}
    For any \(\alpha \!\ge\! 1\), there exists an entrance fee function \(e\) with max-min ratio \(r_e=\alpha\) such that no randomized strategyproof mechanism \(f(e,\cdot)\!:\!\mathbb{R}^n\!\rightarrow\! \Delta(\mathbb{R})\) can achieve a total cost approximation ratio less than \(\frac{\sqrt{2} + 1}{2} - \frac{1}{(4 + 2\sqrt{2}) r_{e} - 2}\).
\end{theorem}
\begin{proof}
    Consider the following entrance fee function:
    \begin{equation*}
          e(\ell) =
          \begin{cases}
            \begin{aligned}
                &d, &&\text{if } \ell \in \{ -1, 1\} \\
                &d + 1-a, &&\text{otherwise} \\
              \end{aligned}
          \end{cases}
    \end{equation*}
    where \(a=\sqrt{2}-1<\frac{1}{2}\).
    Let $r_{e} = \frac{d+1-a}{d}=1+\frac{1-a}{d}$ and \(\gamma^{*} = \frac{\sqrt{2} + 1}{2} - \frac{1}{(4 + \sqrt{2})r_{e} - 2}\).
    For the sake of contradiction, suppose that $f(e,\cdot) \colon \mathbb{R}^n\rightarrow \Delta(\mathbb{R})$ is a randomized strategyproof mechanism and its total cost approximation ratio $\gamma(f(e, \cdot))< \gamma^*$.
    Consider three agent position profiles $\mathbf{x}_{1} = (-1, a)$, $\mathbf{x}_{2} = (-a, 1)$ and $\mathbf{x}_{3} = (-1, 1)$. We illustrate the example in Figure \ref{fig_tc_lb_r3_re}.
    We get \(OPT_{tc}(e,\mathbf{x}_1)=TC(\mathbf{x}_1,-1)=2d+1+a\), \(OPT_{tc}(e,\mathbf{x}_2)=TC(\mathbf{x}_2,1)=2d+1+a\), and \(OPT_{tc}(e,\mathbf{x}_3)=TC(\mathbf{x}_3,\pm 1)=2d+2\).
    Then, we have
    \begin{equation*} \label{B-MC-GAMMA}
        \left\{
        \begin{aligned}
          \gamma(f(e,\mathbf{x}_{1}))
          &= \frac{\mathrm{cost}(-1, f(e, \mathbf{x}_{1})) + \mathrm{cost}(a, f(e, \mathbf{x}_{1}))}{2d + 1+a}; \\
          \gamma(f(e,\mathbf{x}_{2}))
          &= \frac{\mathrm{cost}(-a, f(e, \mathbf{x}_{2})) + \mathrm{cost}(1, f(e, \mathbf{x}_{2}))}{2d + 1+a}; \\
          \gamma(f(e,\mathbf{x}_{3}))
          &= \frac{\mathrm{cost}(-1, f(e, \mathbf{x}_{3})) + \mathrm{cost}(1, f(e, \mathbf{x}_{3}))}{2d + 2}.
        \end{aligned}
        \right.
    \end{equation*}

      \begin{figure} [t]
            \centering
            \includegraphics[width=\myscale\linewidth]{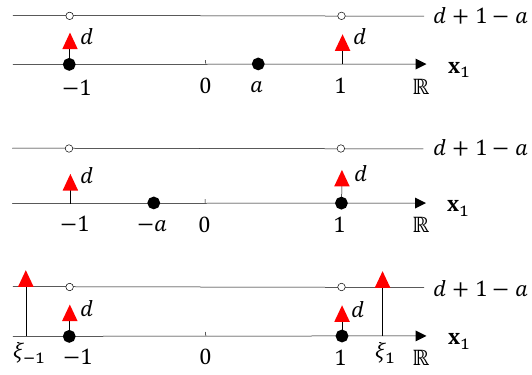}
            \caption{The definitions of \(e(\cdot),\mathbf{x}_1,\mathbf{x}_2, \mathbf{x}_3, \xi_{-1} \textrm{ and } \xi_{1}\). For \(i\in \{1,2,3\}\), the red triangles on \(\mathbf{x}_i\) denote the possible facility locations of \(\hat{f}(e,\mathbf{x}_i)\).} 
            \label{fig_tc_lb_r3_re}
      \end{figure}
    \noindent
    Since $f(\cdot, \cdot)$ is a strategyproof mechanism with approximation ratio less than $\gamma^{*}$, it holds that
    \begin{equation} \label{B-SP-MC}
        \left\{
        \begin{aligned}
          &\gamma(f(e,\mathbf{x}_{i})) < \gamma^{*},~\forall i \in \{ 1, 2, 3 \}; \\
          &\mathrm{cost}(a, f(e, \mathbf{x}_{1})) \leq \mathrm{cost}(a, f(e, \mathbf{x}_{3})); \\
          &\mathrm{cost}(-a, f(e, \mathbf{x}_{2})) \leq \mathrm{cost}(-a, f(e, \mathbf{x}_{3})) .\\
        \end{aligned}
        \right.
    \end{equation}

    Next, we construct a randomized mechanism $\hat{f}(e, \cdot)$ base on $f(e, \cdot)$.
    For any agent position profile \(\mathbf{x}\notin \{\mathbf{x}_1,\mathbf{x}_2,\mathbf{x}_3\}\), let \(\hat{f}(e,\mathbf{x})=f(e,\mathbf{x})\) .
    For \(i\in \{1,2,3\}\), \(\hat{f}(e,\mathbf{x}_i)\) is a discrete probability distribution over finitely many locations, and meanwhile, \(\hat{f}(e,\mathbf{x}_i)\) satisfies \Cref{B-SP-MC} as \(f(e,\mathbf{x}_i)\).
    Intuitively, \(\hat{f}(e,\mathbf{x}_1)\) ``concentrates'' the probability of \(f(e,\mathbf{x}_1)\) over \(\mathbb{R}\setminus \{-1\}\) in \(\{1\}\) and \(\hat{f}(e,\mathbf{x}_2)\) ``concentrates'' the probability of \(f(e,\mathbf{x}_2)\) over \(\mathbb{R}\setminus \{1\}\) in \(\{-1\}\).
    Formally, for $\mathbf{x}_{1}$ and $\mathbf{x}_{2}$, define
    \begin{equation*}
        \Pr[\hat{f}(e, \mathbf{x}_{1}) = \ell] = \left\{
        \begin{aligned}
          &\Pr[f(e, \mathbf{x}_{1}) = -1], && \text{if}~\ell = -1 \\
          &1 - \Pr[f(e, \mathbf{x}_{1}) = -1], && \text{if}~\ell = 1 \\
          &0, && \text{otherwise}
        \end{aligned}
        \right.,
    \end{equation*}
    and
    \begin{equation*}
        \Pr[\hat{f}(e, \mathbf{x}_{2}) = \ell] = \left\{
        \begin{aligned}
          &\Pr[f(e, \mathbf{x}_{2}) = 1], && \text{if}~\ell = 1 \\
          &1 - \Pr[f(e, \mathbf{x}_{2}) = 1], && \text{if}~\ell = -1 \\
          &0, && \text{otherwise}
        \end{aligned}
        \right..
    \end{equation*}
    Thus, \(\hat{f}(e, \mathbf{x}_{1})\) and \(\hat{f}(e, \mathbf{x}_{2})\) are two discrete distributions over \(\{-1,1\}\).
    Observe that when \(\ell\ne -1\), \(TC(\mathbf{x}_1,\ell)\) and \(\mathrm{cost}(a,\ell)\) are both minimized at \(\ell=1\), and when \(\ell \ne 1\), \(TC(\mathbf{x}_2,\ell)\) and \(\mathrm{cost}(-a,\ell)\) are both minimized at \(\ell=-1\).    Thus, for \(i\in \{1,2\}\), \(\hat{f}(e,\mathbf{x}_i)\) performs better than \(f(e,\mathbf{x}_i)\).
    That is, it holds
    \begin{equation} \label{B-M-SP1}
        \left\{
        \begin{aligned}
          &\gamma(\hat{f}(e, \mathbf{x}_{i})) \leq \gamma(f(e,\mathbf{x}_{i})) < \gamma^{*}; \\
          &\mathrm{cost}(a, \hat{f}(e, \mathbf{x}_{1}))
          \leq \mathrm{cost}(a, f(e, \mathbf{x}_{1}))
          \leq \mathrm{cost}(a, f(e, \mathbf{x}_{3}));\\
          &\mathrm{cost}(-a, \hat{f}(e, \mathbf{x}_{2}))
          \leq \mathrm{cost}(-a, f(e, \mathbf{x}_{2}))
          \leq \mathrm{cost}(-a, f(e, \mathbf{x}_{3})). \\
        \end{aligned}
        \right.
    \end{equation}

    Let \(\delta\) be a small positive.
    To construct $\hat{f}(e,\mathbf{x}_{3})$, we first define an auxiliary distribution $g(e, \mathbf{x}_{3})$ that ``concentrates'' the probability of \(f(e,x_3)\) over \((-1,0)\) and \([0,1)\) in \(-1-\delta\) and \(1+\delta\), respectively. Formally, 
    \begin{equation*}
        \Pr[g(e, \mathbf{x}_{3}) = \ell] = \left\{
        \begin{aligned}
          &0, && \text{if}~\ell \in (-1, 1); \\
          &\Pr[f(e, \mathbf{x}_{3}) \in (-1, 0) \cup \{ -1 - \delta \} ], && \text{if}~\ell = -1 - \delta; \\
          &\Pr[f(e, \mathbf{x}_{3}) \in [0, 1) \cup \{ 1 + \delta \}], && \text{if}~\ell = 1 + \delta ;\\
          &\Pr[f(e, \mathbf{x}_{3}) = \ell], && \text{otherwise}.
        \end{aligned}
        \right.
    \end{equation*}
    Thus, \(g(e,\mathbf{x}_3)\) is a probability distribution over \((-\infty,-1]\cup [1,+\infty)\).
    \(g(e,\mathbf{x}_3)\) performs worse than \(f(e,\mathbf{x}_3)\). That is, for $i \in \{1,2\}$ it holds that
    \begin{equation} \label{B-A-SP1}
        \left\{
        \begin{aligned}
          &\mathrm{cost}(a, f(e, \mathbf{x}_{3}))
          \leq \mathrm{cost}(a, g(e, \mathbf{x}_{3})); \\
          &\mathrm{cost}(-a, f(e, \!\mathbf{x}_{3}))
          \leq \mathrm{cost}(-a, g(e, \!\mathbf{x}_{3})). \\
        \end{aligned}
        \right.
    \end{equation}
    For \(\ell\in(-1,0)\), we have \(TC(\mathbf{x}_3,-1-\delta)=TC(\mathbf{x}_3,\ell)+2\delta\). For \(\ell\in [0,1)\), we have \(TC(\mathbf{x}_3,1+\delta)=TC(\mathbf{x}_3,\ell)+2\delta\).
    Thus, for small enough $\delta$, it holds
    \begin{equation} \label{B-A-SP2}
        \gamma(f(e, \mathbf{x}_{3})) < \gamma(g(e,\mathbf{x}_{3})) < \gamma^{*}.
    \end{equation}
    Let $U_{-} = (-\infty, -1)$, and $U_{+} = (1, +\infty)$.
    According to the first mean value theorem for definite integrals, we know that for each $j \in \{-, +\}$, there exists $\xi_{j} \in U_{j}$ such that
    \begin{align*}
        \int_{U_{j}} \Pr[g(e, \textbf{x}_{3}) = \ell]\!\cdot\! |\ell| \,\mathrm{d} \ell
        = |\xi_{j}|\!\cdot\! \int_{U_{j}} \Pr[g(e, \textbf{x}_{3}) = \ell] \,\mathrm{d} \ell
        = |\xi_{j}|\cdot \Pr[g(e, \textbf{x}_{3}) \in U_{j}].
    \end{align*}
    Thus, for any $\tau \in \{-1, -a, a, 1\}$, we have
    \begin{equation} \label{B-FMVT-TC}
        \left\{
        \begin{aligned}
          \int_{U_{-}}\!\!\!\!\!\Pr[g(e, \textbf{x}_{3}) \!=\! \ell] \!\cdot\! |\ell \!-\! \tau| \mathrm{d} \ell
          &\!=\! |\xi_{-1} - \tau| \cdot \Pr[g(e, \textbf{x}_{3}) \in U_{-}]; \\
          \int_{U_{+}}\!\!\Pr[g(e, \textbf{x}_{3}) \!=\! \ell] \!\cdot\! |\ell \!-\! \tau| \mathrm{d} \ell
          &\!=\! |\xi_{1} - \tau| \cdot \Pr[g(e, \textbf{x}_{3}) \in U_{+}].
        \end{aligned}
        \right.
    \end{equation}
    Let
    \begin{equation*}
        \Pr[\hat{f}(e, \mathbf{x}_{3}) = \ell] = \left\{
        \begin{aligned}
          &\Pr[g(e, \mathbf{x}_{3}) = \ell] = \Pr[f(e, \mathbf{x}_{3}) = \ell], && \text{if}~\ell \pm 1; \\
          &\Pr[g(e, \mathbf{x}_{3}) \in U_{-}], && \text{if}~\ell = \xi_{-1}; \\
          &\Pr[g(e, \mathbf{x}_{3}) \in U_{+}], && \text{if}~\ell = \xi_{1}; \\
          &0, && \text{otherwise}. \\
        \end{aligned}
        \right.
    \end{equation*}
    Thus, \(\hat{f}(e, \mathbf{x}_{3})\) is a discrete distribution over \(\{-\xi_{-1},-1,1,\xi_{1}\}\).
    By \Cref{B-FMVT-TC}, for any $\tau \in \{-1, -a, a, 1\}$, it holds that \(\mathrm{cost}(\tau, \hat{f}(e, \textbf{x}_{3})) =  \mathrm{cost}(\tau, g(e, \textbf{x}_{3})).\)
    By \Cref{B-M-SP1,B-A-SP1,B-A-SP2}, we have
    \begin{equation} \label{B-M-SP2}
      \left\{
        \begin{aligned}
          &\gamma(\hat{f}(e, \textbf{x}_{3})) = \gamma(g(e, \textbf{x}_{3})) < \gamma^*;\\
          &\mathrm{cost}(a, \hat{f}(e, \mathbf{x}_{1})) \leq \mathrm{cost}(a, \hat{f}(e, \mathbf{x}_{3})); \\
          &\mathrm{cost}(-a, \hat{f}(e, \mathbf{x}_{2})) \leq \mathrm{cost}(-a, \hat{f}(e, \mathbf{x}_{3})). \\
        \end{aligned}
        \right.
    \end{equation}
    By \Cref{B-M-SP1,B-M-SP2}, we know that $\hat{f}(\cdot, \cdot)$ satisfies the conditions in \Cref{B-SP-MC} as \(f(\cdot,\cdot)\), i.e., 
       \begin{numcases}{}
            \gamma(\hat{f}(e,\mathbf{x}_{i})) < \gamma^{*},~\forall i \in \{ 1, 2, 3 \} \nonumber ; \\
            \mathrm{cost}(a, \hat{f}(e, \mathbf{x}_{1})) \leq \mathrm{cost}(a, \hat{f}(e, \mathbf{x}_{3}))  \label{c1_ge_c3};\\
            \mathrm{cost}(-a, \hat{f}(e, \mathbf{x}_{2})) \leq \mathrm{cost}(-a, \hat{f}(e, \mathbf{x}_{3})) \label{c2_ge_c3}.
        \end{numcases}    

    With the above setup, we proceed to derive a lower bound of \(\gamma(\hat{f}(e,\cdot))\).
    For simplicity, let $p_{i}(\ell) = \Pr[\hat{f}(e, \mathbf{x}_{i}) = \ell]$ for index $i \in \{ 1,2,3\}$.
    Thus, for the agent at \(a\) in \(\mathbf{x}_1\) and the agent at \(-a\) in \(\mathbf{x}_2\), we have
    \begin{equation*}
        \begin{aligned}
          \mathrm{cost}(a, \hat{f}(e, \mathbf{x}_{1})) &= p_{1}(-1)\cdot (d + 1+a) + p_{1}(1)\cdot (d + 1-a)
          = d + 1 + a\cdot (1 - 2p_{1}(1)); \\
          \mathrm{cost}(-a, \hat{f}(e, \mathbf{x}_{2})) &= p_{2}(-1)\cdot (d + 1-a) + p_{2}(1)\cdot (d + 1 + a)
          = d + 1 + a\cdot (1 - 2p_{2}(-1)).
        \end{aligned}
     \end{equation*}
     If the agent at \(a\) in \(\mathbf{x}_1\) deviates to \(1\), her cost would change to
    \begin{align*}
      &\mathrm{cost}(a, \hat{f}(e, \mathbf{x}_{3})) \\
      &=
      p_{3}(\xi_{-1})\cdot (d+1 - \xi_{-1}) +
      p_{3}(\xi_{1})\cdot (d+1-2a + \xi_{1}) + 
      p_{3}(-1)\cdot (d + 1 + a) + p_{3}(1)\cdot (d + 1 - a) \\
      &=
      d+1-\xi_{-1}\cdot p_3(\xi_{-1})+(\xi_1-2a)\cdot p_3(\xi_1)+a\cdot(p_3(-1)-p_3(1)).
    \end{align*}
    By \Cref{c1_ge_c3}, we have
    \begin{align} \label{c3_c1_ge_0}
        \begin{split}
            &\mathrm{cost}(a,\hat{f}(e,\mathbf{x}_3))-\mathrm{cost}(a,\hat{f}(e,\mathbf{x}_1))\\
            &=-\xi_{-1}\cdot p_3(\xi_{-1})+(\xi_1-2a)\cdot p_3(\xi_1)+a\cdot (p_3(-1)-p_3(1)-1+2p_1(1))
            \ge 0.
        \end{split}
    \end{align}
    Similarly, if the agent at \(-a\) in \(\mathbf{x}_2\) deviates to \(-1\), her cost would change to
    \begin{equation*}
      \begin{aligned}
        &\mathrm{cost}(-a, \hat{f}(e, \mathbf{x}_{3}))\\
        &= p_{3}(\xi_{-1})\cdot (d+1 - 2a - \xi_{-1}) + p_{3}(\xi_{1})\cdot (d+1 + \xi_{1}) + 
        p_{3}(-1)\cdot (d + 1 - a) + p_{3}(1)\cdot (d + 1 + a) \\
        &= d+1-(2a+\xi_{-1}+1)\cdot p_3(\xi_{-1})+\xi_{1}\cdot p_3(\xi_1)-a\cdot (p_3(-1)-p_3(1)).
      \end{aligned}
   \end{equation*}
   By \Cref{c2_ge_c3}, we have
   \begin{equation} \label{c3_c2_ge_0}
    \begin{aligned}
        &\mathrm{cost}(-a,\hat{f}(e,\mathbf{x}_3))-\mathrm{cost}(-a,\hat{f}(e,\mathbf{x}_2))\\
        &=-(\xi_{-1}+2a)\cdot p_3(\xi_{-1})+\xi_{1}\cdot p_3(\xi_1)-a\cdot (p_3(-1)-p_3(1)+1-2p_2(-1))
        \ge 0.
    \end{aligned}
    \end{equation}

    \noindent
    Next we consider the approximation ratio of \(\hat{f}(e, \cdot)\) for \(\mathbf{x}_1, \mathbf{x}_2\) and \(\mathbf{x}_3\). We have
    \begin{align}
        \gamma(\hat{f}(e,\mathbf{x}_1))&=p_1(-1)+p_1(1)\cdot \frac{2d+3-a}{2d+1+a} = 1 + p_1(1)\cdot \frac{2-2a}{2d+1+a} < \gamma^* \label{r1_rs}; \\
        \gamma(\hat{f}(e,\mathbf{x}_2))&=p_2(1)+p_2(-1)\cdot \frac{2d+3-a}{2d+1+a} = 1 + p_2(-1)\cdot \frac{2-2a}{2d+1+a}< \gamma^* \label{r2_rs}; \\
        \begin{split}
            \gamma(\hat{f}(e,\mathbf{x}_3))&=p_3(\xi_{-1}) \frac{d+1-a-\xi_{-1}}{d+1} \!+\! p_3(\xi_{1}) \frac{d+1-a+\xi_{1}}{d+1}+p_3(-1)+p_3(1)\\
            &=1+p_3(\xi_{-1})\cdot \frac{-a-\xi_{-1}}{d+1}+p_3(\xi_1)\cdot \frac{-a+\xi_1}{d+1} < \gamma^* \label{r3_rs}.
        \end{split}
    \end{align}
    Thus, summation of \Cref{r1_rs,r2_rs} gives us
    \begin{align}
        \gamma^*>1+\frac{(1-a)\cdot (p_1(1)+p_2(-1))}{2d+1+a} \label{rs_lb_1}.
    \end{align}
    Then, by \Cref{c3_c1_ge_0,r3_rs}, we can derive:
    \begin{align}
        \gamma^* > 1+a\cdot \frac{-p_3(\xi_{-1})+p_3(\xi_1)-p_3(-1)+p_3(1)+1-2p_1(1)}{d+1} \label{rs_lb_2_1}.
    \end{align}
    Similarly, by \Cref{c3_c2_ge_0,r3_rs}, we can obtain:
    \begin{align}
        \gamma^* > 1+a\cdot \frac{p_3(\xi_{-1})-p_3(\xi_{1})+p_3(-1)-p_3(1)+1-2p_2(-1)}{d+1} \label{rs_lb_2_2}.
    \end{align}
    Adding up \Cref{rs_lb_2_1,rs_lb_2_2}, we get
    \begin{align}
        \gamma^* > 1+ \frac{a\cdot (1-p_1(1)-p_2(-1))}{d+1} \label{rs_lb_2}.
    \end{align}
    Combining \Cref{rs_lb_1,rs_lb_2}, we establish a lower bound for \(\gamma^*\):
    \begin{align}
        \gamma^* > 1+\max\bigl(\frac{(1-a)\cdot (p_1(1)+p_2(-1))}{2d+1+a}, \frac{a\cdot (1-p_1(1)-p_2(-1))}{d+1}\bigr) \label{rs_lb_3}.
    \end{align}
    The right hand side of \Cref{rs_lb_3} is minimized when \(\frac{(1-a)\cdot  (p_1(1)+p_2(-1))}{2d+1+a} = \frac{a\cdot (1-p_1(1)-p_2(-1))}{d+1},\) 
    which  gives us \(p_1(1)+p_2(-1)=\frac{a(2d+1+a)}{ad+d+a^2+1}\in [0,1]\). Recall that \(r_e=1+\frac{1-a}{d}\) and \(a=\sqrt{2}-1\). Simple computation shows that the right hand side of \Cref{rs_lb_3} is minimized to \(\frac{\sqrt{2} + 1}{2} - \frac{1}{(4 + 2\sqrt{2}) r_{e} - 2}\), 
    contradicting the assumption that \(\gamma^* = \frac{\sqrt{2} + 1}{2} - \frac{1}{(4 + 2\sqrt{2}) r_{e} - 2}\) in the beginning of the proof.

\end{proof}

\section{Egalitarian Version with One Facility} \label{sect_1f_mc}
Now we study strategyproof mechanisms that approximate the maximum cost.
Our mechanism is \(m_1(\cdot, \cdot)\), which outputs \(x_1^*\) for a given function \(e\) and a location profile \(\mathbf{x}=(x_1,\dots, x_n)\).

\begin{theorem} \label{thm_1f_mc_apx}
    For each entrance fee function $e$ with max-min ratio \(r_e\),
    the approximation ratio of \(m_1(e,\cdot)\) \ for the maximum cost is
        \begin{align*}
            \gamma(m_1(e,\cdot)) \le
            \begin{cases}
                2, &\text{ if } r_e \le 2 \\
                3-\frac{2}{r_e}, &\text{ if } r_e > 2 \\
            \end{cases}.
        \end{align*}
    Hence, the approximation ratio of \(m_{1}(\cdot,\cdot)\) \  for the maximum cost is at most \ \(3\).
\end{theorem}

\begin{proof}
    Given \(\mathbf{x}\in \mathbb{R}^n\). Recall that \(\ell_{mc}\) is the location that achieves the optimal maximum cost. Then \(OPT_{mc}(e,\mathbf{x})=\max\{\mathrm{cost}(x_1,\ell_{mc}),\mathrm{cost}(x_n,\ell_{mc})\}\).
    Let \(t\in \{1,n\}\) be the agent such that \(\mathrm{cost}(x_t,\ell_{mc})=OPT_{mc}(e,\mathbf{x})\). Let \(\ell_{cen}=\frac{x_1+x_n}{2}\) be the center of \(\mathbf{x}\).
    Assume a virtual facility is located at \(\ell_{cen}\) with entrance fee \(E\) such that \(OPT_{mc}(e,\mathbf{x})=\mathrm{vcost}(x_t,\ell_{cen},E)=\frac{L}{2}+E\). Then we have \(E=|\ell_{cen}-\ell_{mc}|+e(\ell_{mc})\) and \(E\ge e(\ell_{mc}) \ge e_{\mathrm{min}}\). Let $L=|x_n-x_1|$.
    Since \(MC(\mathbf{x},f(e,\mathbf{x}))\le L+C_1\), we obtain that
    \begin{align}
        \gamma(m_1(e,\mathbf{x}))\le \frac{L+C_1}{\frac{L}{2}+E}=2+\frac{C_1-2E}{\frac{L}{2}+E}. \label{mc_1}
    \end{align}
    Besides, due to \(C_1\le \mathrm{cost}(x_1, \ell_{mc})\le \mathrm{cost}(x_t, \ell_{mc})\), we have \(L\ge 2C_1-2E\),  and \(E\ge C_1-\frac{L}{2}\) equivalently.

    If \(r_e\le 2\), because \(C_1\le e(x_1)\le e_{\mathrm{max}}\) and  \(E\ge e_{\mathrm{min}}\), we have \(\frac{C_1}{E}\le r_e \le 2\). Then, by \Cref{mc_1}, we get \(\gamma(m_1(e,\cdot))\le 2\).

    For the case that \(r_e> 2\), we divide all agent position profiles into two subcases based on  \(L\) and \(C_1\).

    \noindent
    \textbf{Subcase 1:} \(L< C_1\).
    Considering \(E\ge C_1-\frac{L}{2}\), by \Cref{mc_1}, we get \(\gamma(m_1(e,\mathbf{x}))\le {(L+C_1)}/{(\frac{L}{2}+C_1-\frac{L}{2})}\le 2\).



    \noindent
    \textbf{Subcase 2:} \(L \ge C_1\).
    If \(C_1< 2E\), we have \(\gamma(m_1(e,\mathbf{x}))<2\) by \Cref{mc_1}.
    If \(C_1\ge 2E\), then \(L>2C_1-2E\ge C_1\). Let \(L=2C_1-2E\) in \Cref{mc_1}, and then we get
    \(            \gamma(m_1(e,\mathbf{x}))
    \le 3-2E/C_1.\)
    Since \(\frac{C_1}{E} \le r_e\), we get \(\gamma(m_1(e,\mathbf{x})) \le 3-2/r_e.\)

    When \(r_e>2\), we have \((3-2/r_e) > 2\).
    Thus we have \(\gamma(m_1(e,\mathbf{x})) \le 3-2/r_e\) if \(r_e>2\).
    Together with the case that \(r_e\le 2\), we get the desired result.

\end{proof}
\noindent
We show in the appendix that the approximation ratio in this theorem is tight for the mechanism  \(m_1(e,\cdot)\).

Procaccia et al. \cite{procaccia2009approximate} proved that
when \(e(\cdot)=0\), no deterministic strategyproof mechanism can achieve a maximum cost approximation ratio smaller than \(2\).
We extend this bound to a wider range of entrance fee functions in next theorem.

\begin{theorem} \label{thm_1f_mc_lb_d_2}
    For any given \(\alpha\ge 1\), there exists an entrance fee function \(e\) with max-min ratio \(r_e=\alpha\) such that no deterministic strategyproof mechanism \(f(e,\cdot):\mathbb{R}^n\rightarrow \mathbb{R}\) can achieve an approximation ratio for the maximum cost less than \(2\). This also implies that no deterministic strategyproof mechanism \(f(\cdot,\cdot)\!:\!\mathcal{E}(\alpha)\times \mathbb{R}^n\!\rightarrow\! \mathbb{R}\) can achieve an approximation ratio for the maximum cost less than \(2\).
\end{theorem}
\begin{proof}
    Consider the following entrance fee function  $e$ with  max-min ration $\alpha\geq 1$:
    \begin{align*}
        e(\ell) =
        \begin{cases}
              1, \text{ if } \ell = 1-\alpha \\
              \alpha, \text{ otherwise}
        \end{cases}.
    \end{align*}
    Assume that \(f(\cdot,\cdot):\mathbb{R}^n\rightarrow \mathbb{R}\) is strategyproof.
    Assume for contradiction that, for any \(e \in \mathcal{E}(\alpha)\), \(\gamma(f(e,\cdot))\le 2-\epsilon\) for a certain \(\epsilon\in (0,1)\).
    Consider the \(2\)-agent position profile \(\mathbf{x}=(0, x)\) where \(x> 0\) is a variable.
\begin{figure} [t]
      \centering
      \includegraphics[width=\myscale\linewidth]{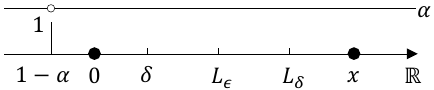}
      \caption{The Definitions of \(e, x, L_{\epsilon},\delta \text{ and } L_{\delta}\)} 
      \label{fig_mc_lb}
\end{figure}
    To ease the proof we also use \(f_x\) to represent \(f(e,(0,x))\).
    Let \(L_{\epsilon}\) be the length of \(\mathbf{x}\) such that \(\frac{L_{\epsilon}+\alpha}{\frac{L_{\epsilon}}{2} + \alpha} = 2-\epsilon,\) and then \(L_{\epsilon}=2(\frac{1}{\epsilon}-1)\alpha\).
    Thus, when \(x> L_{\epsilon}\) we have \(\frac{x+\alpha}{\frac{x}{2}+\alpha}>2-\epsilon\). Note that \(x+\alpha\) is the maximum cost when the facility is at \(0\text{ or } x\) and \(\frac{x}{2}+\alpha\) is the optimal maximum cost for \(\mathbf{x}\). Therefore,
    for any \(x>L_{\epsilon}\), to guarantee \(\gamma(f(e,\mathbf{x}))\le 2-\epsilon\), the facility must be placed in \((0, x)\). That is,  \(f_x\in (0,x)\) for any \(x>L_{\epsilon}\). The example is illustrated in Figure \ref{fig_mc_lb}.

    Next we prove \(f_{x}\in (0,L_{\epsilon}]\) for any \(x>L_{\epsilon}\).
    Suppose for contradiction that there is an \(x'>L_{\epsilon}\) such that \(f_{x'}> L_{\epsilon}\). We then consider the facility location profile \((0,f_{x'})\).  To guarantee \(\gamma(f(e,(0,f_{x'})))\le 2-\epsilon\), we have \(f((0,f_{x'}))\in (0, f_{x'})\). Therefore \(\mathrm{cost}(f_{x'},f(0,f_{x'}))>\alpha\). Then the second agent in profile \((0, f_{x'})\) can deviate to \(x'\) and then the facility will be located at \(f_{x'}\), and thus, her cost is reduced to \(\alpha\). This contradicts the strategyproofness of \(f(\cdot, \cdot)\).

    Furthermore, we prove that \(f_y = f_z\) for any distinct \(y,z > L_{\epsilon}\).
    That is, the facility is fixed in a location in \((0, L_{\epsilon})\) when \(x>L_{\epsilon}\).
    Assume w.l.o.g. that  \(y\ne z\). Suppose for contradiction that \(f_y\ne f_z\).
    If \(f_y> f_z\), then agent in \(z\) can benefit by misreporting \(y\);
    If \(f_y< f_z\), then agent in \(y\) can benefit by misreporting \(z\).
    Thus \(f_y=f_z\). Let \(\delta\in (0, L_{\epsilon}]\) be the value of \(f_{x}\) for all \(x>L_{\epsilon}\). Let \(L_{\delta}\) be the length of \(\mathbf{x}\) such that \(\frac{L_{\delta}-\delta+\alpha}{\frac{L_{\delta}}{2} + \alpha} = 2-\epsilon,\) 
    which solves to \(L_{\delta}=L_{\epsilon}+2\delta/\epsilon\). Thus we have \(\delta \le L_{\epsilon} < L_{\delta}\), as shown in Figure \ref{fig_mc_lb}.
    However, when \(x > L_{\delta}\), we have \(\gamma(f(e,\mathbf{x}))>2-\epsilon\), a contradiction.
  
\end{proof}

By Theorems \ref{thm_1f_mc_apx} and \ref{thm_1f_mc_lb_d_2}, we can conclude that for entrance fee functions with max-min ratio \(\alpha \le 2\), no strategyproof mechanism \(f(\cdot,\cdot)\!:\!\mathcal{E}(\alpha)\times \mathbb{R}^n\!\rightarrow\! \mathbb{R}\) can do better than \(m_1(\cdot,\cdot)\).
Next, when \(\alpha\ge 6\), we establish a tighter lower bound for the maximum cost approximation ratio against all strategyproof mechanisms \(f(\cdot, \cdot)\!:\!\mathcal{E}(\alpha) \times \mathbb R^n \!\rightarrow\! \mathbb R\).

\begin{theorem} \label{thm_1f_mc_d_lb_3re}
    For any \(\alpha \!\ge\! 1\), there exists an entrance fee function \(e\) with max-min ratio \(r_e=\alpha\) such that no deterministic strategyproof mechanism \(f(e,\cdot)\!:\!\mathbb{R}^n\!\rightarrow\! \mathbb{R}\) can achieve a maximum cost approximation ratio less than \(3-\frac{28}{\sqrt{r_e^{2}+20r_e-12}+r_e+10}.\) 
    This also implies that no deterministic strategyproof mechanism \(f(\cdot, \cdot)\!:\!\mathcal{E} \times \mathbb R^n \!\rightarrow\! \mathbb R\) can achieve a maximum cost approximation ratio less than $3$.
\end{theorem}
\begin{proof}
    Let \(d\) be a positive and \(D= d + 3 + 4(d + 1)^{-1}\).
    Consider the following entrance fee function
    \begin{equation*}
        e(\ell) = \left\{
        \begin{aligned}
            &d, &&\text{if } \ell \in \{-1, 1\} \\
            &D, &&\text{otherwise}
        \end{aligned}.
        \right.
    \end{equation*}
    Let $f(e,\cdot)$ be a deterministic strategyproof mechanism.
    Let $\mathbf{x}_{1} = (-\epsilon, 2)$, $\mathbf{x}_{2} = (-2, \epsilon)$ and $\mathbf{x}_{3} = (-2, 2)$ be three agent position profiles, where $\epsilon \in (0, 1)$ is a small positive number.
    We illustrate the example in Figure \ref{fig_mc_lb_d3}.
\begin{figure} [t]
      \centering
      \includegraphics[width=\myscale\linewidth]{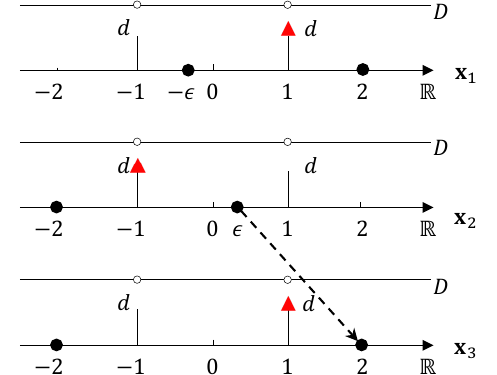}
      \caption{The definitions of \(e(\cdot),\mathbf{x}_1,\mathbf{x}_2 \textrm{ and }\mathbf{x}_3\). The red triangles denote facilities, and the dashed line represents the deviation of the agent.} 
      \label{fig_mc_lb_d3}
\end{figure}

    First we consider the profile $\mathbf{x}_{1}$, we have
    \begin{equation*}
          MC(\mathbf{x}_{1}, \ell) = \left\{
          \begin{aligned}
            &d + 3, &&\text{if } \ell = -1 \\
            &d + 1 + \epsilon, &&\text{if } \ell = 1 \\
            &D + \max\{|\ell - 2|, |\ell + \epsilon| \}, &&\text{otherwise} \\
          \end{aligned}
          \right..
    \end{equation*}
    Thus, \(OPT_{mc}(e,\mathbf{x}_1)=d+1+\epsilon\).
    Assume for contradiction that \(\gamma(f(e, \cdot)) < \frac{d + 3}{d + 1 + \epsilon}.\) 
    Thus, it holds that
    \begin{align*}
        MC(\mathbf{x}_1,f(e,\mathbf{x}_1)) &\le \gamma(f(e, \cdot))\cdot OPT_{mc}(e,\mathbf{x}_1) 
        < d+3 
        < D+ \max(|\ell - 2|, |\ell + \epsilon|).
    \end{align*}
    Then \(MC(\mathbf{x}_1,f(e,\mathbf{x}_1))\) must be \(d+1+\epsilon\). Thus, $f(e,\mathbf{x}_{1}) = 1$ and $\mathrm{cost}(-\epsilon, f(e, \mathbf{x}_{1})) = d + 1 + \epsilon$.
    By the symmetry between profiles $\mathbf{x}_{1}$ and $\mathbf{x}_{2}$, we can get $f(e, \mathbf{x}_{2}) = -1$ and $\mathrm{cost}(\epsilon, f(e, \mathbf{x}_{2})) = d + 1 + \epsilon$.

    Now consider profile $\mathbf{x}_{3}$. We get
    \begin{equation*}
          MC(\mathbf{x}_{3}, \ell) = \left\{
          \begin{aligned}
            &d + 3, &&\text{if } \ell =\pm 1 \\
            &D + \max(|\ell - 2|, |\ell + 2|), &&\text{otherwise} \\
          \end{aligned}
          \right..
    \end{equation*}
    By our assumption and the definition of \(D\), we have \(\gamma(f(e, \cdot)) < \frac{d + 3}{d + 1} = \frac{D + 2}{d + 3},\) 
    which implies that $f(e, \mathbf{x}_{3}) = \pm 1$.
    If $f(e,\mathbf{x}_{3}) = 1$, then the agent at $\epsilon$ in $\mathbf{x}_{2}$ can deviate to $2$ and decrease her cost from $d + 1 + \epsilon$ to $d + 1 - \epsilon$, a contradiction.
    If $f(e,\mathbf{x}_{3}) = -1$, then the agent at $-\epsilon$ in $\mathbf{x}_{1}$ can deviate to $-2$ and decrease her cost from $d + 1 + \epsilon$ to $d + 1 - \epsilon$, a contradiction.
    Due to the arbitrariness of $\epsilon > 0$, we know that the total cost approximation ratio is no less than $\frac{d+3}{d+1}$.
    Observing that $r_{e} = 1 + \frac{3d + 7}{d^2 + d}$, we can get the desired result after simple computation.

\end{proof}

\noindent
\textbf{Lower bound for randomized mechanisms.}
We also obtain the following lower bounds for randomized strategyproof mechanisms.

\begin{theorem} \label{thm_1f_mc_r_lb_2}
    There exists an entrance fee function \(e\) with \(r_e=+\infty\) such that no randomized strategyproof mechanism \(f(e, \cdot): \mathbb{R}^n \rightarrow \Delta(\mathbb{R})\) can achieve a maximum cost approximation ratio less than \(2\).
    Hence, no randomized strategyproof mechanism \(f(\cdot, \cdot):\mathcal{E} \times \mathbb R^n \rightarrow \Delta(\mathbb R)\) can achieve a maximum cost approximation ratio less than $2$.
\end{theorem}
\begin{proof}
    Consider the following entrance fee function
    \begin{equation*}
        e(\ell) =
        \begin{cases}
          \begin{aligned}
              &0, &&\text{if } \ell \in \{ -1, 1\} \\
              &+\infty, &&\text{otherwise} \\
            \end{aligned}
        \end{cases}.
    \end{equation*}
    Let $f(e, \cdot) \colon \mathbb{R}^n\rightarrow \Delta(\mathbb{R})$ be a randomized strategyproof mechanism.
    To achieve a bounded maximum cost approximation ratio, for any $\mathbf{x} \in \mathbb{R}^{n}$ we have \(\Pr[f(e,\mathbf{x}) = -1] + \Pr[f(e, \mathbf{x}) = 1] = 1.\) 
    Let $\mathbf{x}_{1} = (-\epsilon, 2)$, $\mathbf{x}_{2} = (-2, \epsilon)$ and $\mathbf{x}_{3} = (-2, 2)$ be three agent position profiles, where $\epsilon \in (0, 1)$ is a small positive number.
    Let $\Pr[f(e,\mathbf{x}_1) = - 1] = p$ and $\Pr[f(e,\mathbf{x}_2) = 1] = p'$.
    We illustrate the example in Figure \ref{fig_mc_lb_r2}.
    \begin{figure} [t]
        \centering
        \includegraphics[width=\myscale\linewidth]{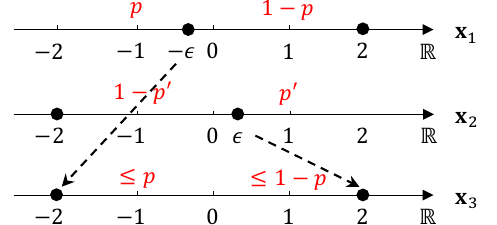}
        \caption{Example for randomized lower bound \(2\) for the maximum cost. The black dots are agents. The red letters denote the probabilities to locate the facility by the mechanism \(f(e,\cdot)\). The dashed lines represent the deviations of the agents.} 
        \label{fig_mc_lb_r2}
\end{figure}
    \(MC(\mathbf{x}_1,1)=1+\epsilon\) is the optimal maximum cost and \(MC(\mathbf{x}_1,-\epsilon)=2+\epsilon\).
    Thus the approximation ratio of $f$ for $\mathbf{x}_{1}$ is
    \begin{equation} \label{mc_bound_1}
        \gamma(f(e,\mathbf{x}_{1})) = (1-p)\cdot 1 + p\cdot \frac{3}{1 + \epsilon}.
    \end{equation}

    Now consider the case that the agent at $-\epsilon$ in $\mathbf{x}_{1}$ deviates to $- 2$.
    Since this agent prefers location $-1$ to location $1$, by the strategyproofness of $f(e,\cdot)$, we must have \(p \geq \Pr[f(e, \mathbf{x}_{3}) = -1]\).
    Then consider the case that the agent at $\epsilon$ in $\mathbf{x}_2$ deviates to $2$. Again by the strategyproofness of $f(e,\cdot)$, we must have \( p' \geq \Pr[f(e, \mathbf{x}_{3}) = 1]\).
    By $\Pr[f(e,\mathbf{x}_3)=-1]+\Pr[f(e,\mathbf{x}_3)=1]=1$, we have $ p+p'\ge 1$.
    W.l.o.g., we assume $p \ge \frac{1}{2}$.
    By \Cref{mc_bound_1}, we have \(\gamma(f(e,\cdot))\ge \gamma(f(e,\mathbf{x}_1)) \geq \frac{2}{1+\epsilon}.\) 
    By the arbitrariness of \(\epsilon>0\), we know that the maximum cost approximation ratio of any randomized strategyproof mechanism is at least \(2\).
  
\end{proof}

Using similar techniques in the proof of Theorem \ref{thm_1f_tc_r_lb_re}, we can establish a lower bound against all the entrance fee functions of a given max-min ratio in the next theorem. The proof is deferred to the appendix.

\begin{theorem} \label{thm_1f_mc_r_lb_re}
    For any \(\alpha \!\ge\! 1\), there exists an entrance fee function \(e\) with max-min ratio \(r_e=\alpha\) such that no randomized strategyproof mechanism \(f(e,\cdot)\!:\!\mathbb{R}^n\!\rightarrow\! \Delta(\mathbb{R})\) can achieve a maximum cost approximation ratio less than \(\frac{24\sqrt{2}+27}{47} - \frac{4}{(60\sqrt{2} + 97)r_{e} - (54\sqrt{2} + 92)}\).
\end{theorem}

\section{Mechanisms with Two Facilities} \label{sect_2f}
In this section, we investigate two-facility games. 
A mechanism is now a function \(f(\cdot,\cdot):\mathcal{E}\times \mathbb{R}^n\rightarrow \mathbb{R}^2\), that given \(e\) and $\mathbf{x}$, returns
the facility location profile \(\bm{\ell}\in \mathbb{R}^2\). Recall for \(\bm{\ell}=(\ell_1,\ell_2)\), the agent \(i\) selects the facility with the smallest sum of travel and entrance fees. Thus, \(\mathrm{cost}(x_i,\bm{\ell})=\min\bigl(|\ell_1 - x_i|+ e(\ell_1), |\ell_2 - x_i|+ e(\ell_2)\bigr)\).

\subsection{Utilitarian Version with Two Facilities}

Let \(i,j\in N\) with \(i\le j\).
Denote by \(m_{i,j}(\cdot,\cdot):\mathbb{R}^n\rightarrow \mathbb{R}^2\) the mechanism that puts one facility in \(x_i^*\) and the other facility in \(x_j^*\) for any input $e$ and $\mathbf{x}$.
Below we show that \(m_{i,j}(\cdot,\cdot)\) is group strategyproof.

\begin{proposition} \label{prop_2f_mij_gsp}
    The mechanism \(m_{i,j}(\cdot,\cdot)\) is group strategyproof.
\end{proposition}
\begin{proof}
    Obviously, agents \(i\) and \(j\) have no incentive to misreport. Since \(x_i\le x_i\), to change the output of the mechanism, there are only three cases to consider: (1) one agent \(k\) on the left side of \(x_i\) misreports a location \(x_k'> x_i\); (2) one agent \(k\) between \(x_i\) and \(x_j\) misreports a location \(x_k'< x_i\) or \(x_k > x_j\); (3) one agent \(k\) on the right side of \(x_j\) misreports a location \(x_k'< x_j\). By \Cref{lm_xs_monotone} and Lemma \ref{lm_3_monotone}, it is easy to show that the cost of agent \(k\) will not decrease in each of the three cases.
\end{proof}

When \(e(\cdot)=0\), \citet{fotakis2014power} prove that any deterministic strategyproof mechanism with bounded approximation ratio (for either of the objectives) must put facilities on \(x_1\) and \(x_n\).
Thus we only consider the mechanism \(m_{1,n}(\cdot,\cdot)\). The following result shows that the total cost approximation ratio of \(m_{1,n}(\cdot,\cdot)\) matches the lower bound of \(n-2\) in the classical model.

\begin{proposition} \label{prop_2f_tc_apx}
    The approximation ratio of \(m_{1,n}(\cdot,\cdot)\) for the total cost is at most \(n-1\).
\end{proposition}
\begin{proof}
    For any entrance fee function \(e\) and any agent position profile \(\mathbf{x}\),
    we define an entrance fee function \(e'\) and an agent position profile \(\mathbf{x}'\) as follows:
    \begin{itemize}
        \item \(\mathbf{x}'\): there are \(n-2\) agents allocated at position \((x_1+x_n)/2\). The other two agents are at positions \(x_1\) and \(x_n\), respectively.
        \item \(e'\): \(e'(x)=e_{\mathrm{min}}\) if \(x=(x_1+x_n)/2\), and \(e'(x)=e_{\mathrm{max}}\) if \(x\ne (x_1+x_n)/2\).
    \end{itemize}

    Next, we show that \(\gamma(m_{1,n}(e,\mathbf{x})) \le \gamma(m_{1,n}(e',\mathbf{x}'))\). Let \((\ell_1, \ell_2)\) be the optimal facility location profile for the total cost. Assume that \(\ell_1\le \ell_2\). Let \(\mathcal{L}\) be the set of agents who select \(\ell_1\) and \(\mathcal{R}\) be the set of agents who select \(\ell_2\). Then,  by Lemma 8 we have \(\mathcal{L}\) and \(\mathcal{R}\) are continuous. Thus, \(e'\) and \(\mathbf{x}'\) can be constructed from \(e\) and \(\mathbf{x}\) in the following steps:
    \begin{enumerate}
        \item Set the entrance fee at \(\ell_1\) and \(\ell_2\) as \(e_{\mathrm{min}}\);
        \item Move agents in \(\mathcal{L}\setminus \{1\}\) to \(\ell_1\) and agents in \(\mathcal{R} \setminus \{1\}\) to \(\ell_2\);
        \item Set the entrance fee at \(\frac{x_1+x_n}{2}\) as \(e_{\mathrm{min}}\) and the entrance fees at \(x_1^*\) and \(x_n^*\) as \(e_{\mathrm{max}}\), and we get \(e'(\cdot)\).
        \item Move agents at \(\ell_1\) to  \(\frac{x_1+x_n}{2}\) and move agents at \(\ell_2\) to \(\frac{x_1+x_n}{2}\), we get \(\mathbf{x}'\).
    \end{enumerate}
    It is easy to see that in each above step, the approximation ratio \(\gamma(m_{1,n}(e,\mathbf{x}))\) will not decrease. Thus, when \(L(\mathbf{x}')=|x_n-x_1|\) is large enough, the optimal total cost for \(\mathbf{x}'\) and \(e'\) is \(L(\mathbf{x}')/2 + (n-1)e_{\mathrm{min}}+ e_{\mathrm{max}}\). \(TC(\mathbf{x}', m_{1,n}(e,\mathbf{x}'))\) is maximized when \(e_{max}<{L(\mathbf{x}')}/{2}+e_{min}\) and \(m_{1,n}(e,\mathbf{x}')=(x_1,x_n)\). Thus, we have
    \begin{align*}
        \gamma(m_{1,n}(e,\mathbf{x})) 
        &\le \frac{TC(\mathbf{x}', m_{1,n}(e,\mathbf{x}'))}{\frac{L(\mathbf{x}')}{2} + (n-1)e_{\mathrm{min}}+ e_{\mathrm{max}}} 
        = \frac{(n-2)\frac{L(\mathbf{x}')}{2}+n\cdot e_{\mathrm{max}}}{\frac{L(\mathbf{x}')}{2}+(n-1)e_{\mathrm{min}}+e_{\mathrm{max}}} \\
        &=\frac{(n-1)\frac{L(\mathbf{x}')}{2}+(n-1)\cdot e_{\mathrm{max}}+e_{\mathrm{max}}-\frac{L(\mathbf{x}')}{2}}{\frac{L(\mathbf{x}')}{2}+(n-1)e_{\mathrm{min}}+e_{\mathrm{max}}}\\
        &\le \frac{(n-1)\frac{L(\mathbf{x}')}{2}+(n-1)\cdot e_{\mathrm{max}}+e_{\mathrm{min}}}{\frac{L(\mathbf{x}')}{2}+(n-1)e_{\mathrm{min}}+e_{\mathrm{max}}} 
        \le n-1.
    \end{align*}

\end{proof}

\noindent
\textbf{Lower bounds for deterministic mechanisms.} The following proposition obtains a similar lower bound against all deterministic strategyproof mechanisms as in Theorem \ref{thm_1f_tc_d_lb_3re}. The basic idea of the proof is to duplicate the agent profiles in the proof of Theorem \ref{thm_1f_tc_d_lb_3re} at locations far enough from the original profiles and then apply the proof of Theorem \ref{thm_1f_tc_d_lb_3re} in a slightly modified way.

\begin{proposition} \label{prop_2f_tc_d_lb_3re}
    For any given \(\alpha> 1\), there exists an entrance fee function \(e\) with max-min ratio \(r_e=\alpha\) such that no deterministic strategyproof mechanism \(f(e,\cdot):\mathbb{R}^n\rightarrow \mathbb{R}^2\) can achieve an approximation ratio less than \(3 - \frac{24}{r_{e} + 7 + \sqrt{r_{e}^{2} + 14r_{e} + 1}}\) for the total cost.
    Hence, no deterministic strategyproof mechanism \(f(\cdot, \cdot):\mathcal{E} \times \mathbb R^n \rightarrow \mathbb R^2\) can achieve a total cost approximation ratio less than $3$.
\end{proposition}
\begin{proof}
    Let \(d>0\) and \(D = d + 2 + (2d+1)^{-1}\).
    Consider the following entrance fee function
    \begin{equation*}
            e(\ell) = \left\{
            \begin{aligned}
                &d, &&\text{if } \ell \in \{-1-L, 1-L,-1, 1\} \\
                &D, &&\text{otherwise}
            \end{aligned}
            \right.,
    \end{equation*}
    where \(L=10(D+1)\) is a large real number.
    Let $f(e,\cdot)\colon \mathbb{R}^n\rightarrow \mathbb{R}^2$ be a deterministic strategyproof mechanism.
    Let $\mathbf{x}_{1} = (-\epsilon-L, 1-L, -1, \epsilon)$, $\mathbf{x}_{2} = (-\epsilon-L, 1-L, -\epsilon, 1)$, and $\mathbf{x}_{3} = (-\epsilon-L, 1-L, -1, 1)$ be three agent position profiles where $\epsilon$ is a small positive. We illustrate this example in Figure \ref{fig_2f_tc_lb_d3}. Denote by \((\ell_1, \ell_2)\) the location profile of the two facilities and assume w.l.o.g. that \(\ell_1\le \ell_2\).

	\begin{figure} [t]
		\centering
		\includegraphics[width=\myscale\linewidth]{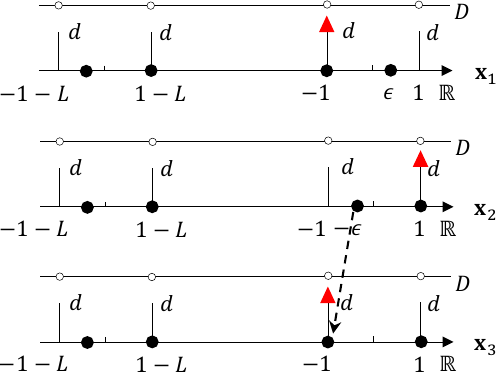}
		\caption{The definitions of \(e(\cdot),\mathbf{x}_1,\mathbf{x}_2 \textrm{ and } \mathbf{x}_3\). The black dots are agents. The red triangles denote facilities, and the dashed line represents the deviation of the agent.} 
		\label{fig_2f_tc_lb_d3}
	\end{figure}

    First, we consider profile $\mathbf{x}_{1}$. We have
    \begin{equation} \label{2f_tc_x1}
          TC(\mathbf{x}_{1}, (\ell_1,\ell_2)) \left\{
          \begin{aligned}
            &= 4d+2+2\epsilon, &&\text{if } (\ell_1,\ell_2) = (1-L, -1) \\
            &= 4d+4, &&\text{if } (\ell_1,\ell_2)\in \{(-1-L,-1),(1-L,  1)\} \\
            &= 4d+6-2\epsilon, &&\text{if } (\ell_1,\ell_2) = (-1-L, 1)\\
            &\ge 4D+2+2\epsilon, &&\text{otherwise} \\
          \end{aligned}
          \right.,
    \end{equation}
    where we have \(4D+2+2\epsilon> 4d +10+2\epsilon\) by our definition of \(D\).
    Thus, the optimal total cost for \(\mathbf{x}_1\) is \(4d+2+2\epsilon\), achieving by \((\ell_1, \ell_2)=(1-L,-1)\).
    Assume, for contradiction, that
    \begin{align} \label{2f_tc_asp_apx}
        \gamma(f(e, \cdot)) < \frac{4d + 6 - 2\epsilon}{4d + 2 + 2\epsilon}
        = \frac{2d+3-\epsilon}{2d+1+\epsilon}
    \end{align}
    Together with \Cref{2f_tc_x1}, we know that \(TC(\mathbf{x}_1,f(e,\mathbf{x}_1))\) must be \(4d+2+2\epsilon\) or \(4d+4\). Then it holds that $f(e,\mathbf{x}_{1})$ locates one of the two facilities at \(1-L\) or \(-1\). Assume, w.l.o.g., that \(-1\in f(e,\mathbf{x}_{1})\).
    Similarly, by our assumption \Cref{2f_tc_asp_apx}, we get $f(e, \mathbf{x}_{2})$ locates one of the two facilities at \(1-L\) or \(1\). Against, w.l.o.g., we assume that \(1\in f(e, \mathbf{x}_{2})\).

    Now consider profile $\mathbf{x}_{3}$.
    We thus have
    \begin{equation} \label{2f_tc_x3}
        TC(\mathbf{x}_{3}, (\ell_1,\ell_2)) \left\{
        \begin{aligned}
          &= 4d+3+\epsilon, &&\text{if } \ell_1=1-L, \ell_2\in \{-1,1\}\\
          &= 4d+5-\epsilon, &&\text{if } \ell_1=-1-L, \ell_2\in \{-1,1\}\\
          &\ge 2(d+D)+3+\epsilon, &&\text{if } \ell_1=1-L, \ell_2\notin \{-1,1\}\\
          &\ge 2(d+D)+5-\epsilon, &&\text{if } \ell_1=-1-L, \ell_2\notin \{-1,1\} \\
          &\ge 4D+3+\epsilon, &&\text{otherwise}\\
        \end{aligned}
        \right..
    \end{equation}
    Thus, the optimal total cost for \(\mathbf{x}_3\) is \(OPT_{tc}(e,\mathbf{x}_3)=4d+3+\epsilon\) when the facilities are located at \((1-L,-1)\) or \((1-L,1)\).
    By our assumption \Cref{2f_tc_asp_apx} and the definition of \(D\), we have \(\gamma(f(e, x_{3})) < \frac{2d + 3}{2d + 1}= \frac{2(d+D)+3}{4d+3}.\) 
    Together with \Cref{2f_tc_x3}, we know that \(TC(\mathbf{x}_3,f(e,\mathbf{x}_1))\) must be \(4d+3+\epsilon\) or \(4d+5-\epsilon\). Then it holds that $f(e,\mathbf{x}_{1})$ locates one of the two facilities at \(-1\) or \(1\). If $-1\in f(e,\mathbf{x}_{3})$, then the agent at $-\epsilon$ in $\mathbf{x}_{2}$ can deviate to $-1$ and decrease her cost from $d + 1 + \epsilon$ to $d + 1 - \epsilon$, a contradiction.
    If $1\in f(e,\mathbf{x}_{3})$, then the agent at $\epsilon$ in $\mathbf{x}_{1}$ can deviate to $1$ and decrease her cost from $d + 1 + \epsilon$ to $d + 1 - \epsilon$, again, a contradiction.
    Due to the arbitrariness of $\epsilon > 0$,
    we have \(\gamma(f(e, \cdot))\ge \frac{2d + 3}{2d + 1}\).
    Observing that $r_{e} = 1 + \frac{4d + 3}{2d^2 + d}$, we can get the desired result after simple computation.

\end{proof}

Note that \citet{fotakis2014power} has proved that when \(e(\cdot)=0\) (or \(r_e=1\)), no deterministic strategyproof mechanisms for two-facility games can achieve an approximation ratio smaller than \(n-2\) for the total cost.
However, to the best of our knowledge, this lower bound cannot be trivially generalized to the case where \(r_e>1\).
In this sense, the lower bound results for \(r_e>1\) in the above proposition can be regarded as complements to the low bound by \citet{fotakis2014power}.

\noindent
\smallskip
\textbf{Lower bounds for randomized mechanisms.} In the classical model where \(r_e=1\), the only known non-trivial (greater than \(1\)) approximation ratio low bound against all randomized strategyproof mechanisms for the total cost is \(1.045\) by \citet{lu2009tighter}.
This lower bound immediately gives the same lower bound of \(1.045\) for any randomized strategyproof mechanism  \(f(\cdot, \cdot):\mathcal{E} \times \mathbb R^n \rightarrow \Delta(\mathbb R^2)\).
In the following proposition, we show that for \(r_e=+\infty\), we can achieve a low bound of \(2\) by extending the proof of Theorem \ref{thm_1f_tc_r_lb_2} to two-facility game. Thus we get a better lower bound against any strategyproof randomized mechanism \(f(\cdot, \cdot):\mathcal{E} \times \mathbb R^n \rightarrow \Delta(\mathbb R^2)\).

\begin{proposition} \label{prop_2f_tc_r_lb_2}
    There exists an entrance fee function \(e\) with \(r_e=+\infty\) such that no randomized strategyproof mechanism \(f(e, \cdot): \mathbb{R}^n\rightarrow \Delta(\mathbb{R}^2)\) can achieve a total cost approximation ratio less than \(2\). Hence, no randomized strategyproof mechanism \(f(\cdot, \cdot):\mathcal{E} \times \mathbb R^n \rightarrow \Delta(\mathbb R^2)\) can achieve a total cost approximation ratio less than $2$.
\end{proposition}

\begin{proof}
    We use the same agent position profiles \(\mathbf{x}_1\), \(\mathbf{x}_2\) and \(\mathbf{x}_3\) as in the proof of Theorem \ref{thm_1f_tc_r_lb_2} for agents \(1\) and \(2\), and  add one more agent (called agent \(0\)) at \(10\).
    We use the same entrance fee function as that in the proof of Theorem \ref{thm_1f_tc_r_lb_2} except that we set the entrance fee at the location of agent \(0\) to \(0\).
    Then to get a total cost approximation ratio less than \(2\) in all three agent position profiles, one of the two facilities must be put at the position of agent \(0\) with probability \(1\), which would incur a cost of \(0\). Then the same arguments in the proof of Theorem \ref{thm_1f_tc_r_lb_2} can be applied to the other facility and agents \(1\) and \(2\) for each agent position profile.

\end{proof}

\subsection{Egalitarian Version with Two Facilities}

Our mechanism for the maximum cost is also \(m_{1,n}(\cdot,\cdot)\). We show that the approximation ratio of \(m_{1,n}(\cdot,\cdot)\) for the maximum cost is the same as \(m_{1}(\cdot,\cdot)\).

\begin{proposition} \label{prop_2f_mc_apx}
    For each entrance fee function $e$ with max-min ratio \(r_e\),
    the approximation ratio of \(m_{1,n}(e,\cdot)\) \ for the maximum cost is
    \begin{align*}
        \gamma(m_{1,n}(e,\cdot)) \le
        \begin{cases}
            2, &\text{ if } r_e \le 2 \\
            3-\frac{2}{r_e}, &\text{ if } r_e > 2 \\
        \end{cases}.
    \end{align*}
    Hence, the approximation ratio of \(m_{1,n}(\cdot,\cdot)\) \  for the maximum cost is at most \ \(3\).
\end{proposition}
\begin{proof}
  Let \(\bm{\ell}=(\ell_{mc}^1, \ell_{mc}^2)\) be the facility location profile that minimizes the maximum cost of \(\mathbf{x}\).
  Let \(\mathbf{x}_L=(x_1, \dots, x_p)\) and  \(\mathbf{x}_R=(x_{p+1}, \dots, x_n)\)  be the sequences of agents that select the facilities at \(\ell_{mc}^1\) and \(\ell_{mc}^2\), respectively. Note that $\mathbf{x}_L$
  and $\mathbf{x}_R$  are correctly defined by Lemma~\ref{lm_continuous}.  Thus, \(\ell_{mc}^1\) is the location that minimizes the maximum cost of \(\mathbf{x}_L\) and \(\ell_{mc}^2\) is the location that minimizes the maximum cost of \(\mathbf{x}_R\). We have two separate one-facility problems. The approximation ratio follows from Theorem \ref{thm_1f_mc_apx}.
\end{proof}

We next obtain a similar lower bound for all deterministic strategyproof mechanisms as in Theorem  \ref{thm_1f_mc_lb_d_2}.

\begin{proposition} \label{prop_2f_mc_lb_2}
    For any \(\alpha\ge 1\), no deterministic strategyproof mechanism $f(\cdot, \cdot)\!:\!\mathcal{E}(\alpha) \times \mathbb R^n \!\rightarrow\! \mathbb R^2$ can achieve a maximum cost approximation ratio less than \(2\).
\end{proposition}
\begin{proof}
    We use the same entrance fee function as in the proof of Theorem \ref{thm_1f_mc_lb_d_2}.
    Let \(D=10(\alpha+1)\) be a large real number.
    Consider the \(4\)-agent position profile \(\mathbf{x}=(-D-x,-D,0,x)\) where \(x>0\) is a variable. Then we can apply the same arguments in the proof of Theorem \ref{thm_1f_mc_lb_d_2} to the agent profiles \((-D-x, -D)\) and \((0,x)\), respectively.
\end{proof}

Recall that Proposition \ref{prop_2f_tc_d_lb_3re} has extended the total cost lower bounds for one-facility games in Theorem \ref{thm_1f_tc_d_lb_3re} to two-facility games by duplicating agent position profiles.
Similarly, we can also extend the maximum cost low bounds for one-facility games in Theorem \ref{thm_1f_mc_d_lb_3re} to two-facility games in the next proposition.

\begin{proposition} \label{prop_2f_mc_d_lb_3re}
    For any \(\alpha \!\ge\! 1\), there exists an entrance fee function \(e\) with max-min ratio \(r_e=\alpha\) such that no deterministic strategyproof mechanism \(f(e,\cdot)\!:\!\mathbb{R}^n\!\rightarrow\! \mathbb{R}^2\) can achieve a maximum cost approximation ratio less than \(3-\frac{28}{\sqrt{r_e^{2}+20r_e-12}+r_e+10}.\) 
    This implies that no deterministic strategyproof mechanism \(f(\cdot, \cdot)\!:\!\mathcal{E} \times \mathbb R^n \!\rightarrow\! \mathbb R^2\) can achieve a maximum cost approximation ratio less than $3$.
\end{proposition}

\begin{proof}
    Let \(d>0\) and \(D = d + 3 + 4(d+1)^{-1}\).
    Consider the following entrance fee function
    \begin{equation*}
            e(\ell) = \left\{
            \begin{aligned}
                &d, &&\text{if } \ell \in \{-1-L, 1-L,-1, 1\} \\
                &D, &&\text{otherwise}
            \end{aligned}
            \right.,
    \end{equation*}
    where \(L=10(D+1)\) is a large real number.
    Let $f(e,\cdot)\colon \mathbb{R}^n\rightarrow \mathbb{R}^2$ be a deterministic strategyproof mechanism.
    Let $\mathbf{x}_{1} = (-\epsilon-L, 2-L, -2, \epsilon)$, $\mathbf{x}_{2} = (-\epsilon-L, 2-L, -\epsilon, 2)$, and $\mathbf{x}_{3} = (-\epsilon-L, 2-L, -2, 2)$ be three agent position profiles where $\epsilon$ is a small positive.
    Let  \(\bm{\ell}=(\ell_1, \ell_2)\) be the location profile of the two facilities and assume w.l.o.g. that \(\ell_1\le \ell_2\).

    First, we consider profile $\mathbf{x}_{1}$. We have
    \begin{equation} \label{2f_mc_x1}
          MC(\mathbf{x}_{1}, \bm{\ell}) \left\{
          \begin{aligned}
            &= d+1+\epsilon, &&\text{if } \bm{\ell} = (1-L, -1) \\
            &= d+3, &&\text{if } \bm{\ell}\in \{(-1-L,-1),(1-L,  1), (-1-L,1)\} \\
            &> D + 1, &&\text{otherwise} \\
          \end{aligned}
          \right.,
    \end{equation}
    where we have \(D+1>d+4\) by our definition of \(D\).
    Thus, the optimal maximum cost for \(\mathbf{x}_1\) is \(d+1+\epsilon\), achieving by \((\ell_1, \ell_2)=(1-L,-1)\).
    Assume, for contradiction, that
    \begin{align} \label{2f_mc_asp_apx}
        \gamma(f(e, \cdot)) < \frac{d + 3}{d + 1 + \epsilon}
    \end{align}
    Together with \Cref{2f_mc_x1}, we know that \(MC(\mathbf{x}_1,f(e,\mathbf{x}_1))\) must be \(d+1+\epsilon\) and \(f(e,\mathbf{x}_1)=(1-L,-1)\).
    Similarly, by our assumption \Cref{2f_mc_asp_apx}, we get $f(e, \mathbf{x}_{2})=(1-L,1)$.

    Now consider profile $\mathbf{x}_{3}$.
    We thus have
    \begin{equation} \label{2f_mc_x3}
        TC(\mathbf{x}_{3}, \bm{\ell}) \left\{
        \begin{aligned}
          &= d+3, &&\text{if } \ell_1\in \{1-L,-1-L\}\}, \ell_2\in \{-1,1\}\\
          &\ge D+2, &&\text{otherwise}\\
        \end{aligned}
        \right..
    \end{equation}
    Thus, the optimal maximum cost for \(\mathbf{x}_3\) is \(d+3\).
    By our assumption \Cref{2f_mc_asp_apx} and the definition of \(D\), we have \(\gamma(f(e, x_{3})) < \frac{d + 3}{d + 1}= \frac{D+2}{d+3}.\) 
    Together with \Cref{2f_mc_x3}, we know that \(TC(\mathbf{x}_3,f(e,\mathbf{x}_1))\) must be \(d+3\). Then it holds that $f(e,\mathbf{x}_{1})$ locates one of the two facilities at \(-1\) or \(1\). If $-1\in f(e,\mathbf{x}_{3})$, then the agent at $-\epsilon$ in $\mathbf{x}_{2}$ can deviate to $-2$ and decrease her cost from $d + 1 + \epsilon$ to $d + 1 - \epsilon$, a contradiction.
    If $1\in f(e,\mathbf{x}_{3})$, then the agent at $\epsilon$ in $\mathbf{x}_{1}$ can deviate to $2$ and decrease her cost from $d + 1 + \epsilon$ to $d + 1 - \epsilon$, again, a contradiction.
    Due to the arbitrariness of $\epsilon > 0$,
    we have \(\gamma(f(e, \cdot))\ge \frac{d + 3}{d + 1}\).
    Observing that $r_{e} = 1 + \frac{3d + 7}{d^2 + d}$, we can get the desired result after simple computation.

\end{proof}

\noindent
\textbf{Lower bounds for randomized mechanisms.} Next, we show that the lower bound of \(2\) in Theorem \ref{thm_1f_mc_r_lb_2} for the one-facility game also holds for the two-facility game.

\begin{proposition} \label{prop_2f_mc_r_lb}
    There exists an entrance fee function \(e\) with \(r_e=+\infty\) such that no randomized strategyproof mechanism \(f(e, \cdot): \mathbb{R}^n \rightarrow \Delta(\mathbb{R}^2)\) can achieve a maximum cost approximation ratio less than \(2\).
    Hence, no randomized strategyproof mechanism \(f(\cdot, \cdot):\mathcal{E} \times \mathbb R^n \rightarrow \Delta(\mathbb R^2)\) can achieve a maximum cost approximation ratio less than $2$.
\end{proposition}
\begin{proof}
    We use the agent position profiles \(\mathbf{x}_1,\mathbf{x}_2\) and \(\mathbf{x}_3\) as that in the proof of Theorem \ref{thm_1f_mc_r_lb_2} for agents \(1\) and \(2\),
    and add one more agent (called agent \(0\)) far away on the left side of agent $1$ for each of the three agent position profiles.
    We use the same entrance fee function as that in the proof of Theorem \ref{thm_1f_mc_r_lb_2} except that we set the entrance fee at the location of agent \(0\) to \(0\).
    Then in order to get an approximation ratio less or equal to \(3\), one of the two facilities must always be located in the position of the new agent \(0\) with probability \(1\). Then the same arguments in Theorem \ref{thm_1f_mc_r_lb_2} can be applied to the second facility and agents \(1\) and \(2\) for each agent position profile.
\end{proof}

\section{Conclusions and Future Work} \label{sect_con}

This paper extends the classical facility location game on the real line by incorporating entrance fee functions, adding versatility to the model. The extension prompts a reevaluation of existing facility location games, like capacitated and heterogeneous facilities, opening avenues for broader applications.
Our arbitrary entrance fee function accommodates non-single-peaked agent preferences, posing new challenges in mechanism design. We address these challenges with straightforward yet robust mechanisms, supported by nearly-tight impossibility results. The proofs for upper and lower bounds employ novel techniques with potential independent interest.

A notable open problem is to narrow the gaps between our bounds in \Cref{tb_our_results}. In the classical model, randomized mechanisms such as the left-right-middle and proportional mechanisms achieve better ratios than deterministic mechanisms. However, these do not extend to our models while remaining strategyproof. Designing improved randomized mechanisms for our model, particularly to close the gap in the two-facility game's total cost, is an intriguing avenue for further exploration.
Another direction for exploration involves proving general structural results, specifically understanding the conditions under which Moulin's characterization of strategyproof mechanisms still holds, with a focus on preferences.

\bibliography{mybibfile}
\newpage
\appendix

\section{Proof of Lemma \ref{lm_ltc_bounded}}
\begin{proof}
    Suppose for contradiction that \(\ell_{tc} < x_1^*\) or \(\ell_{tc} > x_n^*\). We will only disprove \(\ell_{tc} < x_1^*\) since \(\ell_{tc} > x_n^*\) can be disproved by symmetry.

    \noindent
    \textbf{Case 1}: \(x_1^* < x_1\). We have \(\mathrm{cost}(x_1,\ell_{tc}) \geq \mathrm{cost}(x_1,x_1^*)\). Thus \(x_1^*\) dominates \(\ell_{tc}\). This implies \(TC(\bm{x},\ell_{tc}) \geq TC(\bm{x}, x_1^*)\). If \(TC(\bm{x},\ell_{tc}) = TC(\bm{x}, x_1^*)\), by the tie-breaking rule, we have \(\ell_{tc}=x_1^*\), a contradiction. Thus, \(TC(\bm{x},\ell_{tc}) > TC(\bm{x}, x_1^*)\), also a contradiction.

    \noindent
    \textbf{Case 2}: \(x_1^* \ge x_1\). Then we must have  \(\ell_{tc} < x_1\) because of Lemma  \ref{lm_xs_dominate_Ix}. Thus, \(x_1^*\) dominates \(\ell_{tc}\). This implies \(TC(\bm{x},\ell_{tc}) \geq TC(\bm{x}, x_1^*)\). Similarly, this will lead to a contradiction.

    In a similar way, we can prove \(\ell_{mc}\in [x_1^*, x_n^*]\).
\end{proof}

\section{Proof of Lemma \ref{lm_1f_alg}}
\begin{proof}
	Let \(\ell_{tc}^{(k)}\) be the location in \((x_k, x_{k+1}]\) that minimizes the total cost of all agents in \(N\) when the facility location is restricted in \((x_k, x_{k+1}]\). Then
		\begin{align*}
			\ell_{tc}^{(k)} &= \argmin_{\ell\in (x_k, x_{k+1}]} TC(\mathbf{x},\ell) = \argmin_{\ell\in (x_k, x_{k+1}]} (\sum_{i=1}^{i=n} |x_i-\ell|+n\cdot e(\ell)) \\
			& = \argmin_{\ell\in (x_k, x_{k+1}]} (\sum_{i=1}^{i=k} (\ell-x_i) + \sum_{i=k+1}^{i=n} (x_i-\ell) + n\cdot e(\ell))\\
			& = \argmin_{\ell\in (x_k, x_{k+1}]} (n\cdot e(\ell)+(2k-n)\ell-\sum_{i=1}^{i=k}x_i+\sum_{i=k+1}^{i=n}x_i).
		\end{align*}
	By Lemma \ref{lm_ltc_bounded}, it is sufficient to optimize over at most \(n+1\) intervals and we get \(\ell_{tc}\). To obtain \(\ell_{mc}\), we only need to optimize over agents \(1\) and \(n\).
\end{proof}

\section{Proof of Lemma \ref{lm_continuous}}
\begin{proof}
	Let \(\bm{\ell}=\{\ell_1, \dots, \ell_n\}\) be the facility location profile, \(p,q,r\in N\) be three agents such that \(x_p<x_q<x_r\). For agent \(i\in N\), let \(s_i\) be the location of the facility selected by \(i\).
	Assume \(p,r\) select the same facility, say \(\ell_1\), that is, \(s_p=s_r=\ell_1\).
	Let \(M\) be a large positive, say \(n\!\cdot\! \mathrm{TC}(\mathbf{x},\bm{\ell})\).
	Then we construct a new entrance fee function \(e':\mathbb{R}^n\rightarrow \mathbb{R}\) such that \(e'(\ell)=e(\ell_i)\) if \(\ell=\ell_i\) for a certain \(i\in N\) and \(e'(\ell)= M\) if \(\ell \ne \ell_i\) for any \(i\in N\). Then we have that the optimal location for an agent \(i\) under \(e'\) is \(s_i\).
	Denote by \(x_i^*(e')\) be the optimal location for agent \(i\) under entrance fee function \(e'(\cdot)\).
	Thus we have \(x_p^*(e')=x_r^*(e')=\ell_1\).
	Since \(x_p<x_q<x_r\), by \Cref{lm_xs_monotone}, we have \(x_p^*(e')\le x_q^*(e') \le x_r^*(e')\) and then \(x_q^*(e')=\ell_1\).
\end{proof}
        
\section{Tight examples for Theorem \ref{thm_1f_mc_apx}}
\begin{proof}
    When \(r_e\le 2\), the example in Figure \ref{fig_mc_tight_apx_re} shows that the approximation ratio of \(2\) is tight for our mechanism. Let \(e_{\mathrm{max}}\) and \(e_{\mathrm{min}}\) be two positives such that \(e_{\mathrm{max}}\le 2e_{\mathrm{min}}\). Let \(e_{\mathrm{max}}\le 2e_{\mathrm{min}}\). Let \(L> e_{\mathrm{max}}\) be a parameter.
    We consider the following entrance fee function: \(e(\frac{L}{2})=e_{\mathrm{min}}\) and \(e(\ell)=e_{\mathrm{max}}\) for any \(\ell\ne \frac{L}{2}\). There are \(n=2\) agents, and the agent position profile is \(\mathbf{x}=(0,2L)\). We stipulate that \(e_{\mathrm{min}}=e_{\mathrm{max}}-\frac{L}{2}+\epsilon\), where \(\epsilon\) is a small positive.
    In this example, the optimal facility location for the maximum cost is \(\frac{L}{2}\), and the optimal maximum cost is \(L+e_{\mathrm{max}}\). The facility location returned by our mechanism is \(0\) and the maximum cost is \(\frac{L}{2}+e_{\mathrm{min}}\). The approximation ratio is \(\frac{L+e_{\mathrm{max}}}{\frac{L}{2}+e_{\mathrm{min}}}.\) 
    Since \(e_{\mathrm{max}}\le 2e_{\mathrm{min}}\),  this ratio is no greater than \(2\) and approaches to \(2\) when \(L\rightarrow +\infty\).
    
    \begin{figure} [t]
        \centering
        \includegraphics[width=\myscale\linewidth]{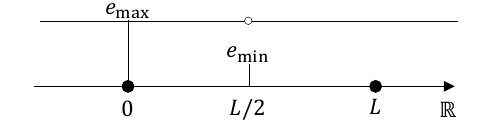}
        \caption{A tight example for the approximation ratio in Theorem \ref{thm_1f_mc_apx}. The black dots are agents. The upper horizontal line represents the entrance fee function \(e\). The height of the vertical line segment represents the entrance fee of the location.} 
        \label{fig_mc_tight_apx_re}
    \end{figure}
    
    When \(r_e>2\), we set \(e_{\mathrm{max}}\) and \(e_{\mathrm{min}}\) in \Cref{fig_mc_tight_apx_re} be such that \(e_{\mathrm{max}}> 2e_{\mathrm{min}}\). 
    Let \(e_{\mathrm{max}}<L<2e_{\mathrm{max}}\).
    We consider the following entrance fee function: \(e(\frac{L}{2})=e_{\mathrm{min}}\) and \(e(\ell)=e_{\mathrm{max}}\) for any \(\ell\ne \frac{L}{2}\). There are \(n=2\) agents, and the agent position profile is \(\mathbf{x}=(0,2L)\). We stipulate that \(e_{\mathrm{min}}=e_{\mathrm{max}}-\frac{L}{2}+\epsilon\), where \(\epsilon\) is a small positive.
    Note that \(e_{\mathrm{max}}<\frac{L}{2}+e_{\mathrm{min}}\) must hold since otherwise the optimal location for agent at \(0\) would be \(\frac{L}{2}\).
    In this example, the optimal facility location for the maximum cost is \(\frac{L}{2}\), and the optimal maximum cost is \(L+e_{\mathrm{max}}\). The facility location returned by our mechanism is \(0\) and the maximum cost is \(\frac{L}{2}+e_{\mathrm{min}}\).
    Recall that \(e_{\mathrm{min}}=e_{\mathrm{max}}-\frac{L}{2}+\epsilon\). The approximation ratio can be verified easily.
\end{proof}

\section{Proof of Theorem \ref{thm_1f_mc_r_lb_re}}
\begin{proof}
    Consider the following entrance fee function
    \begin{align*}
            e(\ell) = \begin{cases}
            \begin{aligned}
            &d, &&\text{if } \ell \in \{ -2, 2 \} \\
            &D = d + 2 - \frac{a}{2}, &&\text{otherwise} \\
            \end{aligned}
        \end{cases},
    \end{align*}
    where $a  = 2\sqrt{2} - 2 \approx 0.828$ is a constant number.
    Let $r_{e} = D/d$ and
    \begin{align*}
        \gamma^*=\frac{24\sqrt{2}+27}{47} - \frac{4}{(60\sqrt{2} + 97)r_{e} - (54\sqrt{2} + 92)}.
    \end{align*}
    Suppose for contradiction that $f(e,\cdot) \colon \mathbb{R}^n\rightarrow \mathbb{R}$ is a randomized strategyproof mechanism with maximum cost approximation ratio less than $\gamma^{*}$. Consider three agent position profiles $\mathbf{x}_{1} = (-2, a)$, $\mathbf{x}_{2} = (-a, 2)$ and $\mathbf{x}_{3} = (-2, 2)$.
    We get \(OPT_{mc}(e,\mathbf{x}_1)=MC(\mathbf{x}_1,-1)=d+1+a\), \(OPT_{mc}(e,\mathbf{x}_2)=MC(\mathbf{x}_2,1)=d+1+a\), and \(OPT_{mc}(e,\mathbf{x}_3)=MC(\mathbf{x}_3,\pm 1)=d+3\).
    Then, we have
    \begin{equation*}
        \begin{cases}
        \begin{aligned}
            \gamma(f(e,\mathbf{x}_{1}))
            &= \frac{\mathbb{E}_{\ell \sim f(e,\mathbf{x}_1)} [\max \bigl(\mathrm{cost}(-2,\ell), \mathrm{cost}(a,\ell)\bigr)]}{d + 1 + a};  \\
            \gamma(f(e,\mathbf{x}_{2}))
            &= \frac{\mathbb{E}_{\ell \sim f(e,\mathbf{x}_2)} [\max \bigl(\mathrm{cost}(-a,\ell), \mathrm{cost}(2,\ell)\bigr)]}{d + 1 + a};  \\
            \gamma(f(e,\mathbf{x}_{3}))
            &= \frac{\mathbb{E}_{\ell \sim f(e,\mathbf{x}_3)} [\max \bigl(\mathrm{cost}(-2,\ell), \mathrm{cost}(2,\ell)\bigr)]}{d + 3}.  \\
        \end{aligned}
    \end{cases}
    \end{equation*}
    Since \(f(\cdot, \cdot)\) is a strategyproof mechanism with approximation ratio less than \(\gamma^*\), it holds that
    \begin{equation} \label{mc_f_conditions}
        \begin{cases}
        \begin{aligned}
            &\gamma(f(e,\mathbf{x}_i)) < \gamma^*, ~\forall i\in \{1,2,3\};\\
            &\mathrm{cost}(a, f(e, \mathbf{x}_{1})))
            \leq \mathrm{cost}(a, f(e, \mathbf{x}_{3})); \\
            &\mathrm{cost}(-a, f(e, \mathbf{x}_{2}))
            \leq \mathrm{cost}(-a, f(e, \mathbf{x}_{3})) .\\
        \end{aligned}
    \end{cases}
    \end{equation}

    Next, we construct a randomized mechanism \(\hat{f}(e,\cdot)\) based on \(f(e,\cdot)\).
    For any agent position profile \(\mathbf{x}\notin \{\mathbf{x}_1, \mathbf{x}_2,\mathbf{x}_3\}\), let \(\hat{f}(e,\mathbf{x})=f(e,\mathbf{x})\).
    For \(i\in \{1,2,3\}\), \(\hat{f}(e,\mathbf{x}_i)\) is a discrete probability distribution over finitely many locations, and meanwhile, \(\hat{f}(e,\mathbf{x}_i)\) satisfies \Cref{mc_f_conditions} as \(f(e,\mathbf{x}_i)\). Intuitively, \(\hat{f}(e,\mathbf{x}_1)\) ``concentrates'' the probability of \(f(e,\mathbf{x}_1)\) over \(\mathbb{R}\setminus \{-1\}\) in \(\{1\}\) and  \(\hat{f}(e,\mathbf{x}_2)\) ``concentrates'' the probability of \(f(e,\mathbf{x}_2)\) over \(\mathbb{R}\setminus \{1\}\) in \(\{-1\}\).
    Formally, for \(\mathbf{x}_1\) and \(\mathbf{x}_2\), define
    \begin{equation*}
        \Pr[\hat{f}(e, \mathbf{x}_{1}) = \ell] =
        \begin{cases}
        \begin{aligned}
            &\Pr[f(e, \mathbf{x}_{1}) = -1], && \text{if}~\ell = -1 \\
            &1 - \Pr[f(e, \mathbf{x}_{1}) = -1], && \text{if}~\ell = 1 \\
            &0, && \text{otherwise}
        \end{aligned}
    \end{cases}
    \end{equation*}
    and
    \begin{equation*}
        \Pr[\hat{f}(e, \mathbf{x}_{2}) = \ell] =
        \begin{cases}
        \begin{aligned}
            &\Pr[f(e, \mathbf{x}_{2}) = 1], && \text{if}~\ell = 1 \\
            &1 - \Pr[f(e, \mathbf{x}_{2}) = 1], && \text{if}~\ell = -1 \\
            &0, && \text{otherwise}
        \end{aligned}.
    \end{cases}
    \end{equation*}
    Thus, \(\hat{f}(e,\mathbf{x}_1)\) and \(\hat{f}(e,\mathbf{x}_2)\) are two discrete distributions over \(\{-1,1\}\).
    Observe that when \(\ell \ne -1\), \(MC(\mathbf{x}_1,\ell)\) and \(\mathrm{cost}(a,\ell)\) are both minimized at \(\ell=1\), and when \(\ell \ne 1\), \(MC(\mathbf{x}_2,\ell)\) and \(\mathrm{cost}(-a,\ell)\) are both minimized at \(\ell=-1\).
    Thus, for \(i\in \{1,2\}\), \(\hat{f}(e,\mathbf{x}_i)\) performs better than \(f(e,\mathbf{x}_i)\).
    That is, it holds
    \begin{equation} \label{f_hatf_x1_x2}
        \begin{cases}
        \begin{aligned}
            &\gamma(\hat{f}(e, \mathbf{x}_{i})) \leq \gamma(f(e,\mathbf{x}_{i})) < \gamma^{*}; \\
            &\mathrm{cost}(a, \hat{f}(e, \mathbf{x}_{1}))
            \leq \mathrm{cost}(a, f(e, \mathbf{x}_{1}))
            \leq \mathrm{cost}(a, f(e, \mathbf{x}_{3})); \\
            &\mathrm{cost}(-a, \hat{f}(e, \mathbf{x}_{2}))
            \leq \mathrm{cost}(-a, f(e, \mathbf{x}_{2}))
            \leq \mathrm{cost}(-a, f(e, \mathbf{x}_{3})). \\
        \end{aligned}
    \end{cases}
    \end{equation}

    Now we construct \(\hat{f}(e,\mathbf{x}_3)\).
    Let $U_{-} = (-\infty, -a) \setminus \{-1\}$, $U_{+} = (a, +\infty) \setminus \{1\}$.
    According to the first mean value theorem for definite integrals, for each index $j \in \{-, +\}$, there exists $\xi'_{j} \in U_{j}$ such that
    \begin{equation*}
        \int_{U_{j}} \Pr[f(e,\mathbf{x}_3)=\ell]\cdot |\ell| \,\mathrm{d} \ell =
        |\xi'_{j}| \cdot \int_{U_{j}} \Pr[f(e,\mathbf{x}_3)=\ell] \,\mathrm{d} \ell =
        |\xi'_{j}| \cdot \Pr[\hat{f}(e, \mathbf{x}_{3}) \in U_{j}].
    \end{equation*}
    Thus, for any $\tau_{+}\in \{a, 2 \}$ and any $\tau_{-} \in \{-2, -a\}$, we have
    \begin{equation} \label{mc_int_U}
        \begin{cases}
            \int_{U_{-}} \Pr[f(e,\mathbf{x}_3)=\ell]\cdot |\ell - \tau_{+}| \,\mathrm{d} \ell
            = |\xi_{-}' - \tau_{+}| \Pr[\hat{f}(e, \mathbf{x}_{3}) \in U_{-}] \\
            \int_{U_{+}} \Pr[f(e,\mathbf{x}_3)=\ell]\cdot |\ell - \tau_{-}| \,\mathrm{d} \ell
            = |\xi_{+}' - \tau_{-}| \Pr[\hat{f}(e, \mathbf{x}_{3}) \in U_{+}]
        \end{cases}.
    \end{equation}
    Similarly, let $W_{-} = [-a, 0]$ and $W_{+} = (0, a]$, for $j \in \{-, +\}$ there exists $\eta_{j} \in W_{j}$ satisfying that
    \begin{align} \label{mc_int_W}
            \int_{W_{j}} \Pr[f(e,\mathbf{x}_3)=\ell]\cdot |\ell| \,\mathrm{d} \ell
            = \eta_{j}\cdot \Pr[f(e,\mathbf{x}_3)\in W_j].
    \end{align}
    Let $\delta > 0$ be a small positive such that $\xi_{-}' - \delta \neq -1$ and $\xi_{+}' + \delta \neq 1$, and let $\xi_{-} = \xi'_{-} - \delta$ and $\xi_{+} = \xi'_{+} + \delta$.
    Then we define \(\hat{f}(e,\mathbf{x}_3)\) as
    \begin{equation*}
        \Pr[\hat{f}(e, \mathbf{x}_{3}) = \ell] =
        \begin{cases}
        \begin{aligned}
            &\Pr[f(e, \mathbf{x}_{3}) \in U_{+}], && \text{if}~\ell = \xi_{+} \\
            &\Pr[f(e, \mathbf{x}_{3}) \in U_{-}], && \text{if}~\ell = \xi_{-} \\
            &\Pr[f(e, \mathbf{x}_{3}) = \ell], && \text{if}~\ell = \pm 1 \\
            &\Pr[f(e, \mathbf{x}_{3}) \in [-a, 0] ], && \text{if}~\ell = \eta_{-} \\
            &\Pr[f(e, \mathbf{x}_{3}) \in (0, a]], && \text{if}~\ell = \eta_{+} \\
            &0, && \text{otherwise}
        \end{aligned}
        \end{cases}.
    \end{equation*}
    Thus, \(\hat{f}(e,\mathbf{x}_3)\) is a discrete distribution over \(\{\xi_{-}, -1, \eta_{-}, \eta_{+}, 1, \xi_{+}\}\). According to \Cref{mc_int_U,mc_int_W}, we have that for any \(\tau\in \{-a,a\}\) it holds that \(\int_{-\infty}^{+\infty} |\tau-\ell|\cdot \Pr[\hat{f}(e,\mathbf{x}_3)=\ell]\mathrm{d}\ell \le \delta+\int_{-\infty}^{+\infty} |\tau-\ell|\cdot \Pr[f(e,\mathbf{x}_3)=\ell]\mathrm{d}\ell.\) 
    Considering the entrance fees, we have
    \begin{align}\label{f_hatf_x3_cost}
        \mathrm{cost}(\tau, f(e,\mathbf{x}_3)) < \mathrm{cost}(\tau, \hat{f}(e,\mathbf{x}_3)) \le \delta+\mathrm{cost}(\tau, f(e,\mathbf{x}_3)).
    \end{align}
    Similarly, by \Cref{mc_int_U,mc_int_W}, we have that
    \begin{align*}
        &\int_{-\infty}^{+\infty} \max(|2-\ell|,|2+\ell|)\cdot \Pr[\hat{f}(e,\mathbf{x}_3)=\ell] \, \mathrm{d}\ell \\
        &= \int_{-\infty}^{0}|2-\ell|\cdot \Pr[\hat{f}(e,\mathbf{x}_3)=\ell]\, \mathrm{d}\ell+\int_{0}^{+\infty}|2+\ell|\cdot \Pr[\hat{f}(e,\mathbf{x}_3)=\ell]\, \mathrm{d}\ell\\
        &\le \delta + \int_{-\infty}^{0}|2-\ell|\cdot \Pr[f(e,\mathbf{x}_3)=\ell]\, \mathrm{d}\ell+\int_{0}^{+\infty}|2+\ell|\cdot \Pr[f(e,\mathbf{x}_3)=\ell]\, \mathrm{d}\ell\\
        &= \delta + \int_{-\infty}^{+\infty} \max(|2-\ell|,|2+\ell|)\cdot \Pr[f(e,\mathbf{x}_3)=\ell] \, \mathrm{d}\ell.
    \end{align*}
    Again, considering the entrance fees, we have
    \begin{align} \label{f_hatf_x3_tc}
        MC(\mathbf{x}_3, f(e,\mathbf{x}_3)) < MC(\mathbf{x}_3, \hat{f}(e,\mathbf{x}_3)) \le \delta + MC(\mathbf{x}_3, f(e,\mathbf{x}_3)).
    \end{align}
    By \Cref{f_hatf_x1_x2,f_hatf_x3_cost,f_hatf_x3_tc}, for small enough \(\delta\), we have
    \begin{equation} \label{f_hatf_x3}
            \begin{cases}
                \begin{aligned}
                    &\gamma(f(e,\mathbf{x}_3) < \gamma(\hat{f}(e,\mathbf{x}_3) < \gamma^*; \\
                    &\mathrm{cost}(a,\hat{f}(e,\mathbf{x}_1)\le \mathrm{cost}(a,\hat{f}(e,\mathbf{x}_3)); \\
                    &\mathrm{cost}(-a,\hat{f}(e,\mathbf{x}_2)\le \mathrm{cost}(-a,\hat{f}(e,\mathbf{x}_3)).
                \end{aligned}
            \end{cases}
    \end{equation}
    By \Cref{f_hatf_x1_x2,f_hatf_x3}, we know that \(\hat{f}(\cdot,\cdot)\) satisfies the conditions in \Cref{mc_f_conditions} as \(f(\cdot,\cdot)\), i.e.,
    \begin{numcases}{}
        \gamma(\hat{f}(e,\mathbf{x}_{i})) < \gamma^{*},~\forall i \in \{ 1, 2, 3 \} \nonumber; \\
        \mathrm{cost}(a, \hat{f}(e, \mathbf{x}_{1})) \leq \mathrm{cost}(a, \hat{f}(e, \mathbf{x}_{3}))  \label{mc_c1_ge_c3};\\
        \mathrm{cost}(-a, \hat{f}(e, \mathbf{x}_{2})) \leq \mathrm{cost}(-a, \hat{f}(e, \mathbf{x}_{3})) \label{mc_c2_ge_c3}.
    \end{numcases}
    
    With the above setup, we can now proceed to derive a lower bound of $\gamma(\hat{f}(e,\cdot))$. For the sake of presentation, let $p_{i}(\ell) = \Pr[\hat{f}(e, \mathbf{x}_{i}) = \ell]$ for $i \in \{ 1,2,3 \}$.  Thus, for agent at $a$ in $\mathbf{x}_1$ and agent at $-a$ in $\mathbf{x}_2$, 
        \begin{equation*}
        \begin{aligned}
            &\mathrm{cost}(a, \hat{f}(e, \mathbf{x}_{1})) = p_{1}(-1) \cdot (d + 1 + a) + p_{1}(1) \cdot (d + 1 - a) 
            = d + 1 + a - p_{1}(1) \cdot 2a; \\
            &\mathrm{cost}(-a, \hat{f}(e, \mathbf{x}_{2})) = p_{2}(1) \cdot (d + 1 + a) + p_{2}(-1) \cdot (d + 1 - a) 
            = d + 1 + a - p_{2}(-1) \cdot 2a.
        \end{aligned}
        \end{equation*}
    If the agent at $a$ in $\mathbf{x}_1$ deviates to $1$, her cost would change to
    \begin{align*}
        \mathrm{cost}(a,\hat{f}(e,\mathbf{x}_3))&= p_3(\xi_{-})(a-\xi_{-}+d+2-\frac{a}{2}) + p_3(\xi_{+})(a-\xi_{+}+d+2-\frac{a}{2})+\\
        &\hspace{0.45cm}p_3(-1)(a+1+d) + p_3(1)(1-a+d)+\\
        &\hspace{0.45cm}p(\eta_{-})(a-\eta_{-}+d+2-\frac{a}{2}) + p(\eta_{+}) (a-\eta_{+}+d+2-\frac{a}{2})\\
        &=d+1+a+p_3(\xi_{-})(1-\frac{a}{2}-\xi_{-})+p_3(\xi_{+})(1-\frac{5a}{2}-\xi_{+})+\\
        &\hspace{0.45cm}p_3(\eta_{-})(1-\frac{a}{2}-\eta_{-})+p_3(\eta_{+})(1-\frac{a}{2}-\eta_{+})+p_3(1)\cdot (-2a).
    \end{align*}
    By \Cref{mc_c1_ge_c3}, we have
    \begin{align} \label{mc_c3_c1_ge_0}
        \begin{split}
            \mathrm{cost}(a,\hat{f}(e,\mathbf{x}_3))\!-\!\mathrm{cost}(a,\hat{f}(e,\mathbf{x}_1))
            &= p_3(\xi_{-})(1\!-\!\frac{a}{2}\!-\!\xi_{-})+p_3(\xi_{+})(1\!-\!\frac{5a}{2}\!-\!\xi_{+})+\\
            &\hspace{.45cm}p_3(\eta_{-})(1\!-\!\frac{a}{2}\!-\!\eta_{-})+p_3(\eta_{+})(1\!-\!\frac{a}{2}\!-\!\eta_{+})+\\
            &\hspace{.45cm}p_3(1)\cdot(-2a)+p_1(1)\cdot 2a\\
            &\ge 0.
        \end{split}
    \end{align}
    Similarly, for the agent at \(-a\) in \(\mathbf{x}_2\), we have
    \begin{align} \label{mc_c3_c2_ge_0}
        \begin{split}
            \mathrm{cost}(-a,\hat{f}(e,\mathbf{x}_3))\!-\!\mathrm{cost}(-a,\hat{f}(e,\mathbf{x}_2))
            &\!=\! p_3(\xi_{\!-})(1\!-\!\frac{5a}{2}\!-\!\xi_{\!-})\!+\!p_3(\xi_{\!+})(1\!-\!\frac{a}{2}\!+\!\xi_{\!+})+\!\\
            &\hspace{.35cm}p_3(\eta_{\!-})(1\!-\!\frac{a}{2}\!+\!\eta_{\!-})\!+\!p_3(\eta_{\!+})(1\!-\!\frac{a}{2}\!+\!\eta_{\!+})+\!\\
            &\hspace{.35cm}p_3(-1)\cdot(-2a)+p_2(-1)\cdot 2a\\
            &\!\ge\! 0.
        \end{split}
    \end{align}
    
    Next we consider the approximation ratio of \(\hat{f}(e, \cdot)\) for \(\mathbf{x}_1, \mathbf{x}_2\) and \(\mathbf{x}_3\). We have
    \begin{align}
        &\gamma(\hat{f}(e, \mathrm{x}_{1})) = p_{1}(-1) + p_{1}(1) \cdot \frac{d + 3}{d + 1 + a} = 1 + p_{1}(1) \cdot \frac{2 - a}{d + 1 + a} < \gamma^{*}; \label{mc_r1_rs}\\
        &\gamma(\hat{f}(e, \mathrm{x}_{2})) = p_{2}(1) + p_{2}(-1) \cdot \frac{d + 3}{d + 1 + a} = 1 + p_{2}(-1) \cdot \frac{2 - a}{d + 1 + a} < \gamma^{*} \label{mc_r2_rs};\\
        \begin{split} \nonumber 
            \gamma(\hat{f}(e, \mathrm{x}_{3})) 
            &= p_3(\xi_{-})\cdot \frac{d+2-\frac{a}{2}+2-\xi_{-}}{d+3}+p_3(\xi_{+})\cdot \frac{d+2-\frac{a}{2}+2+\xi_{+}}{d+3}+\\
            &\hspace{.45cm}p_3(\eta_{-})\cdot \frac{d+2-\frac{a}{2}+2-\eta_{-}}{d+3}+p_3(\eta_{+})\cdot \frac{d+2-\frac{a}{2}+2+\eta_{+}}{d+3}+
            p_3(-1)+p_3(1) \\
            &= 1+\frac{p_3(\xi_{-})\cdot (1-\frac{a}{2}-\xi_{-})+p_3(\xi_{+})\cdot(1-\frac{a}{2}+\xi_{+})}{d+3}+\\
            &\hspace{.45cm} \frac{p_3(\eta_{-})\cdot (1-\frac{a}{2}-\eta_{-})+p_3(\eta_{+})\cdot (1-\frac{a}{2}+\eta_{+})}{d+3}\\
            &< \gamma^{*}.\\
        \end{split}
    \end{align}
    Thus, \Cref{mc_r1_rs} \(+\) \Cref{mc_r2_rs} gives
    \begin{equation} \label{mc_rs_lb_1}
        \gamma^* > 1+\frac{(2-a)\cdot (p_1(1)+p_2(-1))}{2(d+1+a)}.
    \end{equation}
    Then, by \Cref{mc_c3_c1_ge_0,mc_r1_rs}, we can derive
    \begin{equation}\label{mc_rs_lb_2_1}
        \gamma^* > 1+\frac{p_3(\xi_{+})\cdot 2a + p_3(\eta_{+})\cdot 2\eta_{+}-p_1(1)\cdot 2a + p_3(1)\cdot 2a}{d+3}.
    \end{equation}
    Similarly, by \Cref{mc_c3_c2_ge_0,mc_r2_rs}, we have
    \begin{equation}\label{mc_rs_lb_2_2}
        \gamma^* > 1+\frac{p_3(\xi_{-})\cdot 2a - p_3(\eta_{-})\cdot 2\eta_{-}-p_2(-1)\cdot 2a + p_3(-1)\cdot 2a}{d+3}.
    \end{equation}
    By adding \Cref{mc_rs_lb_2_1,mc_rs_lb_2_2}, we get
    \begin{align*}
        \gamma^* > 1 + \frac{p_3(\xi_{+}) a + p_3(\xi_{-}) a+p_3(\eta_{+}) \eta_{+} - p_3(\eta_{-}) \eta_{-}+p_3(1) a + p_3(-1)  a}{d+1}
        - \frac{(p_1(1)+p_2(-1)) a}{d+1}.
    \end{align*}
    Since \(p_3(\xi_{-})+p_3(\xi_{+})+p_3(\eta_{-})+p_3(\eta_{+})+p_3(-1)+p_3(1)=1\), we have
    \begin{equation} \label{mc_rs_lb_2}
        \gamma^* > 1 + \frac{a\cdot(1-(p_1(1)+p_2(-1)))}{d+1}.
    \end{equation}
    Combining \Cref{mc_rs_lb_1,mc_rs_lb_2}, we establish a lower bound for \(\gamma^*\):
    \begin{equation} \label{mc_rs_lb_3}
        \gamma^* > 1+\max\Bigl(\frac{(2-a)\cdot (p_1(1)+p_2(-1))}{2(d+1+a)}, \frac{a\cdot(1-(p_1(1)+p_2(-1)))}{d+1}\Bigr).
    \end{equation}
    The right hand side of \Cref{mc_rs_lb_3} is minimized when \(\frac{(2-a)\cdot (p_1(1)+p_2(-1))}{2(d+1+a)} = \frac{a\cdot(1-(p_1(1)+p_2(-1)))}{d+1},\) 
    which  gives us \(p_1(1)+p_2(-1)=\frac{2 a (a + d + 1)}{2 a^{2} + a (d - 1) + 2 (d + 3)}\in [0,1].\) 
    Recall that \(r_e=1+\frac{4-a}{2d}\) and \(a=2\sqrt{2}-2\). 
    Simple computation shows that the right hand side of \Cref{mc_rs_lb_3} is minimized to \(\frac{24\sqrt{2}+27}{47} - \frac{4}{(60\sqrt{2} + 97)r_{e} - (54\sqrt{2} + 92)}\),
    which contradicts the assumption in the beginning of the proof.
\end{proof}
\end{document}